\documentclass[journal,10pt]{IEEEtran}
\usepackage{graphicx, epstopdf}
\usepackage{amssymb, amsmath}
\usepackage{subeqnarray}
\usepackage{cases}
\usepackage{commath}
\usepackage{setspace}
\usepackage{acronym}
\usepackage{mathrsfs}
\usepackage{ifthen}
\usepackage{textcomp}
\usepackage{float}
\usepackage{subfigure}
\usepackage{soul}
\usepackage{color}
\usepackage[dvipsnames]{xcolor}
\usepackage{cite}
\usepackage{setspace}
\usepackage{mathtools}
\usepackage{bm}
\usepackage{url}
\usepackage{xcolor}
\usepackage{empheq}
\usepackage{tikz}
\usepackage{subfigure}
\usepackage{stfloats,amsthm}
\usepackage{threeparttable}
\usepackage{color}
\newtheorem{theorem}{Theorem}
\newtheorem{lemma}{Lemma}
\newtheorem{remark}{Remark}

\DeclareMathSizes{12}{12}{7.7}{5.5}
\usepackage[bb=boondox,bbscaled=1]{mathalfa}

\def\H{{\bf H}}

\def\y{{\bf y}}

\def\0{{\bf 0}}
\def\1{{\bf 1}}

\def\vect{\mathsf{vec}}

\begin{document}

\title{Beyond MMSE: Rank-1 Subspace Channel Estimator for Massive MIMO Systems}

\author{Bin~Li,
        Ziping~Wei,
        Shaoshi~Yang,~\emph{Senior Member, IEEE},
        Yang~Zhang,
        Jun~Zhang,~\emph{Senior Member, IEEE},
        Chenglin~Zhao,
				Sheng Chen,~\emph{Life Fellow, IEEE}

\thanks{This work was supported in part by the Beijing Municipal Natural Science Foundation under Grant L234025 and Grant Z220004, in part by the National Natural Science Foundation of China under Grant 62071247, in part by the Beijing Municipal Science \& Technology Commission under Grant Z221100007722036,  and in part by the Fundamental Research Funds for the Central Universities under Grant 2023ZCJH02. \textit{(Corresponding author: Shaoshi Yang, Chenglin Zhao.)} } %
\thanks{B.~Li, Z.~Wei, S. Yang and C.~Zhao are with School of Information and Communication Engineering, Beijing University of Posts and Telecommunications, Beijing, 100876, China (E-mails: Binli@bupt.edu.cn, weizp@bupt.edu.cn, shaoshi.yang@bupt.edu.cn, clzhao@bupt.edu.cn).}%
\thanks{Y.~Zhang is with China Satellite Network Group Co., Ltd, Beijing, China (E-mail: zhangyang@bupt.cn).}%
\thanks{J.~Zhang is with Jiangsu Key Laboratory of Wireless Communications, Nanjing University of Posts and Telecommunications, Nanjing 210003, China (E-mail: zhangjun@njupt.edu.cn).} %
\thanks{S.~Chen is with School of Electronics and Computer Science, University of Southampton, Southampton SO17 1BJ, U.K. (E-mail: sqc@ecs.soton.ac.uk).} %
\vspace*{-5mm}
}

\maketitle

\begin{abstract}
To glean the benefits offered by massive multi-input multi-output (MIMO) systems, channel state information must be accurately acquired.
Despite the high accuracy, the computational complexity of classical linear minimum mean squared error (MMSE) estimator becomes prohibitively high in the context of massive MIMO, while the other low-complexity methods degrade the estimation accuracy seriously.
In this paper, we develop a novel rank-1 subspace channel estimator {\color{black}to approximate the maximum likelihood (ML) estimator}, which outperforms the linear MMSE estimator, but incurs a surprisingly low computational complexity.
Our method first acquires the highly accurate angle-of-arrival (AoA) information via a constructed space-embedding matrix and the rank-1 subspace method.
Then, it adopts the \emph{post-reception} beamforming to acquire the unbiased estimate of channel gains.
Furthermore, a fast method is designed to implement our new estimator.
Theoretical analysis shows that the extra gain achieved by our method over the linear MMSE estimator grows according to the rule of $\textsf{O}(\log_{10} M)$, while its computational complexity is \emph{linearly} scalable to the number of antennas $M$.
Numerical simulations also validate the theoretical results.
Our new method substantially extends the accuracy-complexity region and constitutes a promising channel estimation solution to the emerging massive MIMO communications.
\end{abstract}


\IEEEpeerreviewmaketitle

\section{Introduction}\label{sec:intro}
Massive multi-input multi-output (MIMO), as well as its evolution,  has been recognized as a key technology in 5G and 6G communications \cite{marzetta2010noncooperative,ngo2013energy,hoydis2013massive,larsson2014massive,boccardi2014five,2015Fifty,extreme_massive_MIMO_6G,Wang_channel_model_6G}. It is capable of substantially improving both the spectrum and energy efficiency \cite{ngo2013energy}, thus enabling a wide variety of emerging applications \cite{larsson2014massive,2016_mMIMO_HSR,2016_massive_MIMO_video}. However, the vast potential benefits of massive MIMO can be attained only if the channel state information (CSI) is accurately acquired.
{\color{black}Given the prohibitive complexity of the maximum likelihood (ML) estimator,} the task of accurate CSI acquisition is usually accomplished by another classical linear minimum mean squared error (MMSE) channel estimator \cite{hoydis2013massive, shariati2014robust, Hu_ESPRIT_mMIMO, LV201630, Zhou_Angle_Estimation, semiblind_mMIMO_CE, FDD_mMIMO_CE}. The linear MMSE estimator exploits the covariance matrix of wireless channels and is capable of attaining high accuracy. Unfortunately, the {\it a priori} statistical covariance information tends to be unavailable, particularly in realistic dynamic environments. Moreover, the linear MMSE estimator may also incur an excessively high computational complexity and a significant processing delay, and hence becomes impractical to low-cost and low-power devices. For these reasons, the linear least-square (LS) channel estimator \cite{marzetta2010noncooperative} has been widely applied. However, the LS estimator inevitably degrades the CSI accuracy and thus the channel capacity, despite its benefit of low-complexity implementation \cite{shariati2014low,Li2020Randomized}.

Recently, a variety of other channel estimators have been proposed for massive MIMO systems, in order to achieve more balanced tradeoff between estimation accuracy and computational complexity. These estimators can be divided into three types. The first type of channel estimators exploit the low-rank property \cite{Xie2017An,shen2015joint,eliasi2017low,Li2020Randomized}, as massive MIMO channels are in general spatially correlated \cite{Xie2017An}. By approximately computing the inverse of a large matrix in the linear MMSE estimator, a weighted polynomial expansion channel (WPEACH) estimator was designed \cite{moshavi1996multistage,shariati2014low}, which results in a substantial reduction in computational complexity. However, in time-varying environments, the updating of the weight coefficients is highly complex \cite{Li2020Randomized}. The authors of \cite{shen2015joint} proposed a singular value projection (SVP) method, which has lower mean squared error (MSE) than the LS estimator and lower complexity than the linear MMSE estimator. The second type of estimators resort to the sparsity property of massive MIMO channels operating in millimeter-wave (mmWave) frequencies \cite{2010Compressed,alkhateeb2014channel,alkhateeby2015compressed,gao2016channel,rodriguezfernandez2018frequency,huang2018asymptotically,huang2019iterative}. For example, an iterative LS estimator with sparse message passing (LSE-SMP) was proposed to acquire CSI \cite{huang2018asymptotically,huang2019iterative}, and it attains the minimum variance unbiased estimate.

The third type of channel estimators first exploit the array signal processing techniques \cite{krim1996two, schmidt1986multiple, Li2019Fast} that are widely used in radar systems, such as multiple signal classification (MUSIC) and fast Fourier transform (FFT), to estimate the angle of arrival (AoA), and then invoke the acquired accurate AoA information to assist the fading gain estimation.
With the accurate information of AoAs, it is expected that the signal-to-noise ratio (SNR) in fading gain estimation will be improved \cite{guo2017millimeter,fan2017angle}.
{\color{black}
When the MUSIC algorithm is applied, the high resolution AoA estimate can be attained,
yet its complexity is formidable to massive MIMO systems.
Besides, the subspace computation is infeasible for uniform linear array (ULA), as the covariance matrix becomes unavailable in the multi-user case.
For this reason, FFT has been widely applied. As a result, this type of channel estimators are capable of  effectively reducing the complexity.
However, due to its leakage effect, the angular resolution is seriously limited and the estimated AoAs  even become \emph{biased}.
Therefore, the existing angular-domain estimators may potentially suffer from either unaffordable complexity or  degraded accuracy.
}

The designs of the state-of-the-art channel estimators for massive MIMO are to strike an attractive compromise between the two conflicting objectives, namely, the computational complexity and the estimation accuracy. Most existing methods can reduce the computational complexity to some extent \cite{shariati2014low,shen2015joint,fan2017angle}, compared with the linear MMSE estimator; but none of them is capable of outperforming the linear MMSE estimator. Despite its great significance in both theoretical analysis and practical applications, designing a channel estimator whose achievable performance is superior to that of the linear MMSE estimator remains an important open problem, especially when the low computation complexity is also required.

Against the above backdrop, in this work we propose a novel rank-1 subspace channel estimator for massive MIMO.
To the best of our knowledge, it is the first time that an approximate ML estimator is designed, which can outperform the classical linear MMSE estimator, yet requiring a dramatically reduced computational complexity.
To attain this, it firstly acquires the AoA information by leveraging a space-embedding Hankel matrix and exploiting the rank-1 subspace, which attains the near-optimal AoA estimates and breaks the resolution limit of the FFT method.
Second, a \emph{post-reception} beamforming scheme is suggested to estimate fading gains, relying on the ML criterion.
More importantly, we provide the theoretical result on the lower-bound error of this two-stage estimator, showing that the achieved gain over the linear MMSE estimator increases obeying the rule of $\textsf{O}(\log_{10} M)$, where $M$ is the number of antennas at the base station (BS).
In order to reduce the complexity, we further integrate a fast method, by using the low rank property of channel matrix of massive MIMO communications.
It thus achieves the highest-ever complexity reduction, which is only on the order of $\textsf{O}\big(K P^2 M\big)$, compared with $\textsf{O}\big(M^3\big)$ for the MMSE estimator, whereby $K < M$ is the number of users and $P\ll M$ is the number of paths.

Note that, although the single-snapshot MUSIC was studied in MIMO radar systems \cite{liao2016music,fortunati2014single,maisto2022single}, this is the first time it is introduced to massive MIMO communications.
By incorporating a post-reception beamforming scheme and the ML-based estimation of fading gains, the single-snapshot MUSIC approximates the ML estimator.
Moreover, we maximally simplify the computational decomposition of large matrix, thus making complex subspace methods applicable to massive MIMO systems.
As such, our method overcomes two challenges in classical CSI estimators, by significantly improving the accuracy and reducing the complexity.
%Thus, the achievable accuracy-complexity region of massive MIMO CSI estimation is extended substantially.
The main contributions of this work are summarized as follows.

\begin{enumerate}
\item We develop a two-stage channel estimator for massive MIMO systems by leveraging the rank-1 subspace method.
We first obtain the estimation of unknown AoAs, relying on a constructed Hankel matrix and an estimated pseudo-spectrum.
Thus, the high-resolution AoAs are acquired, resulting in the estimation error on the order of $\textsf{O}\big(1/M^{1+\varepsilon}\big)$ for $0<\varepsilon<1$.
Moreover, we adopt the post-reception beamforming to estimate unknown channel gains, which effectively improves the SNR and, unlike the classical FFT method, attains the \emph{unbiased} channel gain estimation.

\item We prove that the Cram$\acute{\text{e}}$r-Rao lower bound (CRLB) of our two-stage CSI estimator is improved by $\textsf{O}\big(1/M\big)$, compared with the linear MMSE estimator. For the target MSE, our near-ML estimator thus attains an extra gain that increases according to $\textsf{O}(\log_{10} M)$. As a result, the channel estimation accuracy is dramatically improved, which in turn enhances the communication coverage and the channel capacity.

\item We design a fast method to implement this rank-1 subspace channel estimator, in the case of low-rank channel matrix. By utilizing the inherent low-rank property of a large Hankel matrix, we first approximate it via three small matrices. Then, the singular value decomposition (SVD) of the large Hankel matrix is replaced with the SVD of these small matrices. This reduces the computational complexity of our estimator from $\textsf{O}\big(KM^3\big)$ to $\textsf{O}\big(KP^2M\big)$~$(P\ll M)$. Moreover, we show that our fast method attains the high-resolution AoA estimation, without compromising the accuracy.

\item The theoretical performance analysis of our new estimator is consolidated by comprehensive numerical simulation results, which demonstrate that our method indeed outperforms the linear MMSE estimator and the sparsity-based methods, while dramatically reducing the computational complexity by 2$\thicksim$3 orders of magnitude. In other words, the substantial performance gain is attained by imposing a much lower complexity. Thus, our new CSI estimator substantially extends the attainable complexity-accuracy region of CSI estimation for massive MIMO systems.
\end{enumerate}

The rest of our paper is organized as follows. In Section~\ref{S2}, we analyze the AoA-aided post-reception beamforming for acquiring the channel gains in the single-user case, for emphasizing the critical role of accurate AoA estimation. Section~\ref{S3} details our proposed rank-1 subspace channel estimator, which involves the successive estimation of AoAs and channel gains. In Section~\ref{S4}, we theoretically analyze the estimation accuracy of our method. In Section~\ref{S5} a fast method is developed to implement our estimator, and it remarkably reduces the computational complexity without degrading the accuracy. Numerical simulations are provided in Section~\ref{S6} to validate our method, and we conclude this paper in Section~\ref{S7}.

This paper adopts the following notation conventions. Boldface lower-case and capital letters denote vectors and matrices, respectively. $\bm{X}^{*}$, $\bm{X}^{\rm T}$ and $\bm{X}^{\rm H}$ represent the conjugate, the transpose and the conjugate transpose of $\bm{X}$, respectively, while $\bm{X}^{-1}$ and $\bm{X}^\dag=\bm{X}^{\rm H}\big(\bm{X}\bm{X}^{\rm H}\big)^{-1}$ are the inverse and the Moore-Penrose pseudo-inverse of $\bm{X}$, respectively. $\vect(\bm{H})$ is the vectorization of $\bm{H}$, $\otimes$ is the Kronecker product, and $\lceil \cdot \rceil$ is the ceiling operator. $\bm{I}_M$ is the $M \times M$ identity matrix, and  $\|\cdot\|_F$ is the Frobenius norm, while $\mathsf{E}[\cdot ]$ denotes the expectation operator and $\text{var}[\cdot]$ stands for variance.

\section{Channel Estimation Aided With AoA}\label{S2}
In the uplink training of massive MIMO systems, if the AoA information is available to the BS, it can be utilized in the post-reception beamforming to improve the channel estimation accuracy. In this section, we first analyze the AoA-aided post-reception beamforming for acquiring the channel gains in a single-user case, upon assuming perfect AoA information. This provides some insight on how the post-reception beamforming improves the channel estimation accuracy. We next derive a necessary condition on the estimated AoA information in order to attain the \emph{unbiased} channel estimation, and thereby show that classical FFT-based angular-domain channel estimators may fail to satisfy this condition.

We consider the up-link channel estimation with a single-antenna user equipment (UE), which transmits the pilot signal of length $B\geq 1$, while the BS equipped with $M$ antennas estimates the unknown channel gains. Hence, the received signal $\bm{Y}\in\mathbb{C}^{M\times B}$ is given by
\begin{align} \label{eq1}
\bm{Y} =& \bm{h} \bm{x}^{\rm H} + \bm{N},
\end{align}
where $\bm{x}\in \mathbb{C}^{B\times 1}$ is the random pilot signal (e.g., i.i.d. Gaussian sequence) with a covariance matrix $\sigma_x^2{\mathbf I}_B$, $\bm{N}\in \mathbb{C}^{M\times B}$ is the additive white Gaussian noise (AWGN) matrix whose entries have zero mean and variance $\sigma_n^2$, and $\bm{h}\in \mathbb{C}^{M\times 1}$ is the wireless channel vector. Without loss of generality, we use the Saleh-Valenzuela (S-V) channel model to capture both the spatial sparsity and the low-rank property of massive MIMO channels \cite{Brady2013Beamspace,alkhateeb2014channel,2016Beamspace,Li2013On,li2015a}.
The low-rank channel occurs widely in massive MIMO communications, especially when the BS  deployed in high buildings suppresses local scattering \cite{Xie2017An,shen2015joint} and the signal propagation of multiple clustered users shares the same reflectors.
Then, the channel vector between the UE and the BS is given by \cite{Xie2017An,shen2015joint,2016Beamspace,Li2020Randomized}
\begin{align} \label{eq2}
\bm{h} =& \sum_{p=0}^{P-1} \alpha_p \bm{a}_M\big(\theta_p\big) = \sum_{p=0}^{P-1} \bm{c}_p ,
\end{align}
where $P$ is the number of paths; $\bm{c}_0\triangleq \alpha_0\bm{a}_M\big(\theta_0\big)$ is the line-of-sight (LoS) component with the complex-valued gain $\alpha_0$ and the AoA $\theta_0$; while $\bm{c}_p\triangleq \alpha_p\bm{a}_M\big(\theta_p\big)$, $1\le p\le P-1$, denotes the $p$th non-line-of-sight (NLoS) component that has the complex-valued gain $\alpha_p$ and the AoA $\theta_p$ \cite{gustafson2014on,shen2015joint}. As shown in \cite{akdeniz2014millimeter}, when operating in the high frequency band (e.g., millimeter-wave band), the number of paths typically satisfy $P \leq 5$ in realistic environments. If ULA is used,
%\footnote{Note that we use the ULA for the ease of analysis. As will be demonstrated later, our proposed channel estimator can be directly generalized to two-dimensional (2D) uniform rectangular array (URA).}
the steering vector $\bm{a}_M(\theta )\in \mathbb{C}^{M \times 1}$ corresponding to the AoA $\theta$ reads
\begin{align} \label{eq3}
\bm{a}_M(\theta ) \triangleq & \Big[1 ~ \exp \Big( \tfrac{j 2\pi d \sin\theta}{\lambda} \Big) ~ \cdots ~
\exp \Big( \tfrac{j2\pi (M-1)d \sin\theta}{\lambda} \Big)  \Big]^{\rm T} .
\end{align}
Here, $\lambda$ is the carrier wavelength and $d$ is the antenna spacing.

If we further define the received signal vector as
\begin{align} \label{eq4}
\bm{y} \triangleq & \bm{Y}\bm{x}\in \mathbb{C}^{M\times 1} ,
\end{align}
and assume that the accurate AoA information of $P$ paths, $\big\{\widehat{\theta}_p \big\}_{p=0}^{P-1}$, is available, then we estimate the channel gains by invoking a post-reception beamforming as
\begin{align} \label{eq5}
\widehat{\alpha}_p =& \frac{1}{M} \bm{a}_M^{\rm H}\big(\widehat{\theta}_p\big) \bm{y} , ~ 0\le p\le P-1 .
\end{align}

\begin{theorem} \label{T1}
When the uplink AoAs of a single-antenna UE are exactly known, i.e., $\widehat{\theta}_p=\theta_p$ for $0\le p\le P-1$, and the number of BS antennas $M$ is large, the channel gains estimated by using the post-reception beamforming (\ref{eq5}) are unbiased, i.e., $\mathsf{E}\big[\hat{\alpha}_p\big]=\alpha_p$, with the estimation error bounded by $\frac{\sigma_n^2}{MB\sigma_x^2}$, for $0\le p\le P-1$.
\end{theorem}

\begin{proof}
See Appendix~\ref{ApA}.
\end{proof}
It can be seen that in this idealized scenario with perfect $P$ AoAs, the estimation accuracy for the channel gains approaches the optimum, achieving a CRLB of $\frac{\sigma_n^2}{MB\sigma_x^2}$.
%This compares favourably even with the maximum likelihood (ML) estimator with the unknown AoAs. The CRLB of the ML estimator $\widehat{h}_p=\widehat{\alpha}_p \bm{a}_M\big(\widehat{\theta}_p\big)$ in this case is $\sigma_n^2/B$.
In practice, however, the BS does not have the perfect AoA information and can only estimate these $P$ AoAs. Naturally, the accuracy of the estimated AoAs, $\big\{\widehat{\theta}_p\big\}_{p=0}^{P-1}$, affects the accuracy of estimating the channel gains, when using the post-reception beamforming (\ref{eq5}).

\begin{theorem} \label{T2}
To achieve the unbiased estimation of the channel gains via (\ref{eq5}), a necessary condition is that the estimation error of the AoA information must satisfy:
\begin{align} \label{eq6}
\big|\widehat{\theta}_{p} - \theta_{p}\big| \textcolor{black}{=} & \textsf{O}\big(1/M^{1+\varepsilon}\big ) , \, \exists \, \varepsilon>0, \forall p .
\end{align}
\end{theorem}

\begin{proof}
See Appendix~\ref{ApB}.
\end{proof}
As mentioned, there exist some angular-domain channel estimators, which first acquire the AoAs and then estimate the channel gains \cite{fan2017angle}. In terms of the achievable MSE, such angular-domain estimators \cite{fan2017angle} are superior to the LS estimator, but inferior to the linear MMSE estimator. The main reason is that, in the case of ULA, the information available at the BS is limited and thus inadequate for classical subspace methods to acquire high-resolution AoAs\footnote{Another 2D MUSIC method has been studied in \cite{guo2017millimeter}, which also becomes inapplicable to the case considered here, i.e., one ULA serves multiple UEs, where the rank of the covariance matrix $\bm{R}_{\bm{y}}=\frac{1}{M}\bm{y}\bm{y}^{\rm H}$ is one.}. Thus, only the FFT method can be applicable to acquiring AoAs, but it attains low-resolution and even biased AoAs, due to the limited FFT length (recall that the angular resolution is $\Delta\theta=1/M$ in the FFT method). As a result, the channel gain estimation upon using the post-reception beamforming is biased unfortunately. Then, the FFT-based AoA estimation, with an error decaying rate obeying $\textsf{O}(1/M)$, is insufficient to achieve the unbiased channel estimation via the post-reception beamforming. This will also be demonstrated later by our numerical simulations, where the FFT-based method exhibits the MSE floor in the high-SNR region.

The above discussion emphasizes the significant importance of the viable, high-resolution AoA estimates for the accurate CSI acquisition in massive MIMO systems, which partially motivates our work presented in this paper.

\section{Rank-1 Subspace Channel Estimator}\label{S3}

In this section, we first introduce a new approach to acquire high-resolution AoA information, which then assists us to obtain the unbiased and accurate channel estimation via the post-reception beamforming. In the case of ULA with $K$ UEs, the received signal reduces to an $M\times 1$ spatial-domain vector, since the other dimension of the signal matrix vanishes after removing the multi-user interference (MUI), as will be seen later. Thus, classical subspace methods (e.g. standard MUSIC) relying on the covariance matrix becomes invalid. To overcome this difficulty, we design a novel rank-1 subspace channel estimator, by leveraging the \emph{space-embedding} technique in analogy to the time-embedding technique used in modeling dynamical systems \cite{1996Dynamical,2015Hankelet,georgakis2018dynamic}. Note that the spatial smoothing MUSIC technique is also applicable in this context.

\subsection{Single-User Uplink Channel Estimation}\label{rank_1_esti} % S3.1
We begin with the single-user case ($K=1$). Let the length of pilot signal be $B\geq K= 1$. The elements of the received signal vector $\bm{y}=[y_0 ~ y_1 \cdots y_{M-1}]^{\rm H}$ of (\ref{eq4}) can be expressed for $0\le m\le M-1$ as
\begin{align}\label{eq7}
y_m =& \sum\limits_{p=0}^{P-1} \alpha_p \exp\left( \frac{j2\pi m d \sin\theta_p}{\lambda}\right) + n_m .
\end{align}
It can be seen that the SNR is increased by $B$ times, since we have  $\text{var}\big[n_m\big]=\sigma_n^2/B\sigma_x^2$. By defining the steering matrix as
\begin{align}\label{eq8}
\bm{A}_M(\bm{\theta}) =& \left[ \bm{a}_M(\theta_0) ~ \bm{a}_M(\theta_1) ~ \cdots ~ \bm{a}_M(\theta_{P-1}) \right] \in \mathbb{C}^{M \times P},
\end{align}
with $\bm{\theta}=[\theta_0 ~ \theta_1 \cdots \theta_{P-1}]^{\rm H}$, the received signal vector can be restructured as
\begin{align}\label{eq9}
\bm{y} =& \bm{A}_M(\bm{\theta}) \bm{\alpha} + \bm{n},
\end{align}
where we have $\bm{\alpha}=\big[\alpha_0 ~ \alpha_1 \cdots \alpha_{P-1}\big]^{\rm H} \in \mathbb{C}^{P \times 1}$.

To extract the unknown AoA information $\bm{\theta}$, we adopt a novel rank-1 subspace method which allows for high-resolution AoA estimation \cite{hacker2010single,liao2016music}. To be specific, we first construct a space-embedding Hankel matrix $\bar{\bm{Y}}\in \mathbb{C}^{L\times (M-L)}$ from the received signal vector $\bm{y}$, i.e.,
\begin{align}\label{eq:rank_limited} % eq10
 \bar{\bm{Y}}\! \triangleq &\! \left[\!\! \begin{array}{cccc}
  y_0    \!\! &\! y_1   \! &\! \cdots\! &\! y_{M-L-1} \\
  y_1    \!\! &\! y_2   \! &\! \cdots\! &\! y_{M-L} \\
  \vdots \!\! &\! \vdots\! &\! \ddots\! &\! \vdots \\
  y_{L-1}\!\! &\! y_{L} \! &\! \cdots\! &\! y_{M-1}
 \end{array}\!\! \right]\!  % \nonumber \\ =&
 \! =\! \bar{\bm{A}}_L \bm{\Lambda} \bar{\bm{A}}_{M-L}^{\rm H}\! +\! \bar{\bm{N}} ,\!
\end{align}
where we have the diagonal matrix $\bm{\Lambda}\triangleq\text{diag}\{\alpha_0,\alpha_1,\cdots,\alpha_{P-1}\}$, the stack length $L$ satisfies $L\geq P$ and ${M-L}\geq P$, while $\bar{\bm{A}}_L\in \mathbb{C}^{L \times P}$ is the sub-matrix consisting of
the first $L$ rows of $\bm{A}_M(\bm{\theta})$ given in (\ref{eq8}).
\textcolor{black}{Since the above Hankel matrix $\bar{\bm{Y}}$ admits the Vandermonde decomposition, one can directly obtain the signal space and noise subspace by performing SVD on $\bar{\bm{Y}}$ \cite{liao2016music}.}

Thus, like other subspace-based methods, e.g., the MUSIC method \cite{schmidt1986multiple,Li2019Fast,roy1989esprit-estimation}, in order to accomplish  high-resolution AoA estimation we first compute the SVD of $\bar{\bm{Y}}$ as
\begin{align}\label{eq11}
\bar{\bm{Y}} =& \bm{U} \bm{\Sigma} \bm{V}^{\rm H} = \bm{U}_P \bm{\Sigma}_P \bm{V}_P^{\rm H} + \bm{U}_{-P} \pmb{\Sigma}_n  \bm{V}_{-P}^{\rm H},
\end{align}
where we denote the left singular matrix $\bm{U}\in \mathbb{C}^{L\times L}$, the diagonal singular value matrix $\bm{\Sigma}_P=\text{diag}\big\{\sigma_1(\bar{\bm{Y}}),\sigma_2(\bar{\bm{Y}}),\cdots ,\sigma_P(\bar{\bm{Y}})\big\}$ and $\bm{\Sigma}_n=\text{diag}\big\{\sigma_{P+1}(\bar{\bm{Y}}),\sigma_{P+2}(\bar{\bm{Y}}),\cdots ,\sigma_{L}(\bar{\bm{Y}})\big\}$, with the singular values $\sigma_1(\bar{\bm{Y}})\geq \cdots \geq \sigma_P(\bar{\bm{Y}})\ge\sigma_{P+1}(\bar{\bm{Y}})= \cdots =\sigma_L(\bar{\bm{Y}})^2$, and the right singular matrix is $\bm{V}\in\mathbb{C}^{(M-L)\times L}$, while $\bm{U}_P\in \mathbb{C}^{L\times P}$ is the submatrix consisting of the first $P$ columns of $\bm{U}$ and it corresponds to the signal subspace,
%\textcolor{red}{ $\bm{\Sigma}_P=\text{diag}\big\{\sigma_1(\bar{\bm{Y}}),\sigma_2(\bar{\bm{Y}}),\cdots,\sigma_P(\bar{\bm{Y}})\}$},
and $\bm{V}_P\in \mathbb{C}^{(M-L)\times P}$ is the submatrix consisting of the first $P$ columns of $\bm{V}$. Furthermore, $\bm{U}_{-P}\in \mathbb{C}^{L\times (L-P)}$ is the submatrix consisting of the last $(L-P)$ columns of $\bm{U}$ and it corresponds to the noise subspace, and $\bm{V}_{-P}\in \mathbb{C}^{(M-L)\times (L-P)}$ is the submatrix consisting of the last $(L-P)$ columns of $\bm{V}$.
\textcolor{black}{Note that the first $P$ singular values of matrix $\bar{\bm{Y}}$ are significantly larger than the remaining $L-P$ singular values (e.g., in the high-SNR case).
Thus, we refer $\bar{\bm{Y}}$ as an approximate low-rank matrix, i.e., $\text{rank}\big(\bar{\bm{Y}}\big)=\text{rank}(\bm{\Lambda})\approx P$.}

For each angle $\theta \in [-\pi, ~ \pi]$, the spatial pseudo-spectrum $P(\theta)$ is then estimated via \textcolor{black}{\cite{liao2016music}}:
\begin{align}\label{eq12}
P(\theta) =& \frac{1}{\bm{a}_L^{\rm H}(\theta)\big(\bm{I}_L - \bm{U}_{P}\bm{U}_{P}^{\rm H}\big) \bm{a}_L(\theta)},
\end{align}
where $\bm{a}_L(\theta)\in \mathbb{C}^{L \times 1}$ denotes the subvector consisting of the first $L$ elements of $\bm{a}_M(\theta)$. Upon evaluating the spatial pseudo-spectrum with the total $N$ discretized angles $\{-\pi+n \Delta\theta,~ n=0,1,\cdots,N-1\}$, where $\Delta\theta =2\pi/N$ is the angle resolution and $N$ is a sufficiently large integer, multiple AoAs of unknown channel paths can be extracted, e.g., by identifying the $P$ highest peaks in $P(\theta)$:
\begin{align}\label{eq13}
\widehat{\bm{\theta}} \triangleq & \Big\{ \widehat{\theta}_p ~ | ~ \widehat{\theta}_p=\texttt{peak}[P(\theta)] , ~ 0\le p\le P-1 \big\}.
\end{align}

After obtaining the unknown AoAs $\big\{\widehat{\theta}_p\big\}_{p=0}^{P-1}$, the estimated channel gains $\big\{\widehat{\alpha}_p\big\}_{p=0}^{P-1}$ can be computed using the post-reception beamforming (\ref{eq5}).
In fact, based on the {\color{black} obtained highly accurate AoA information} our estimator essentially achieves the ML estimation:
\begin{align} \label{eq:ML_esti} % eq14
\widehat{\alpha}_p =& \arg \max\limits_{\alpha_p\in \mathbb{C}} p\left( \frac{1}{M}\bm{a}_M^{\rm H}\big(\widehat{\theta}_p\big)(\bm{h}-\bm{y})\big| \alpha_p\right) ,
\end{align}
where the likelihood is given by
\begin{align}\label{eq15}
p\left( \frac{1}{M}\bm{a}_M^{\rm H}\big(\widehat{\theta}_p\big) (\bm{h}-\bm{y})\big| \alpha_p\right) \approx & p\left( \alpha_p - \frac{1}{M}\bm{a}_M^{\rm H}\big(\widehat{\theta}_p\big)\bm{y}\right) \nonumber \\
\approx & p\left( \frac{1}{M}\bm{a}_M^{\rm H}\big(\widehat{\theta}_p\big) \bm{n}\right),
\end{align}
since $\frac{1}{M}\bm{a}_M^{\rm H}\big(\widehat{\theta}_p\big)\bm{a}_M^{\rm H}\big(\theta_p\big)\approx 1$ and $\frac{1}{M}\bm{a}_M^{\rm H}\big(\widehat{\theta}_p\big)\bm{a}_M^{\rm H}\big(\theta_{p'}\big)\approx 0$, $p\neq p'$, for large $M$ and accurate $\widehat{\theta}_p$. That is, it is a Gaussian distribution, and solving the optimization problem (\ref{eq:ML_esti}) leads to the ML solution of (\ref{eq5}). Finally, we arrive at the equivalent angular-domain channel estimate
\begin{align}\label{eq16}
\widehat{\bm{h}} =&  \sum\limits_{p=0}^{P-1} \widehat{\alpha}_p \bm{a}_M\big(\widehat{\theta}_p\big) .
\end{align}

\vspace{-0.5cm}
\subsection{Multiple-User Uplink Channel Estimation}\label{S3.2}
Now consider the multi-user MIMO (MU-MIMO) case, where the BS equipped with $M$ antennas serves $K$ single-antenna users. Given the pilot signal $\bm{X}=\big[ \bm{x}_1 ~ \bm{x}_2 \cdots \bm{x}_K\big] \in \mathbb{C}^{B\times K}$, the received signal matrix $\bm{Y}\in\mathbb{C}^{M\times B}$ is given by
\begin{align} \label{eq:linear_gaussian} % eq17
\bm{Y} =& \bm{H} \bm{X}^{\rm H} + \bm{N} ,
\end{align}
where again $\bm{N}$ is the channel AWGN matrix, and the channel matrix $\bm{H}=\big[\bm{h}_1 ~ \bm{h}_2 \cdots \bm{h}_K\big]\in\mathbb{C}^{M\times K}$ with
\begin{align} \label{eq18}
\bm{h}_k =& \sum \limits_{p=0}^{P-1} \alpha_{p,k} \bm{a}_M\big(\theta_{p,k}\big) , ~ 1\le k\le K .
\end{align}

By exploiting the orthogonal property for the pilot sequences of different users, namely, $\bm{x}_k^{\rm H}\bm{x}_k=1$ and $\bm{x}_k^{\rm H}\bm{x}_{k'}=0$ for $k,k'\in\{1,2,\cdots K\}$ and $k\neq k'$, we first obtain the received signal vector of the $k$th user:
\begin{align}\label{eq19}
\bm{y}_k =& \bm{Y} \bm{x}_k \in \mathbb{C}^{M\times 1}, ~ 1\le k\le K.
\end{align}
\textcolor{black}{For each signal vector $\bm{y}_k$, the MUI has been removed.
Thus, we will adopt the rank-1 subspace method in Section III-A to acquire the highly accurate AoAs of user $k$.}
%However, in this case, a covariance matrix $\textbf{R}_y\simeq\bm{y}_k\bm{y}_k^{\rm H}$ of $k$-th received signal vector becomes a rank-1 matrix, i.e., $\text{rank}(\textbf{R}_y)=1$.
%As such, classical subspace methods (e.g. MUSIC) is invalid \cite{liao2016music}.}

To be specific, a space-embedding matrix $\bar{\bm{Y}}_k$ is constructed similar to (\ref{eq:rank_limited}).
On this basis, we can estimate the unknown channel vector of the $k$th user as $\widehat{\bm{h}}_k= \sum_{p=0}^{P-1}\widehat{\alpha}_{p,k} \bm{a}_M(\widehat{\theta}_{p,k})$. By repeating the above single-user estimator for $K$ times, the whole channel matrix can be estimated as $\widehat{\bm{H}}=\big[\widehat{\bm{h}}_1 ~ \widehat{\bm{h}}_2\cdots \widehat{\bm{h}}_k\big]$.

\vspace{-0.3cm}
\subsection{Complexity Analysis}\label{S3.3}
For the single-user case, the computational complexity of our rank-1 subspace channel estimator consists of three parts: 1)~the SVD of a space-embedding matrix $\bar{\bm{Y}} \in \mathbb{C}^{L\times (M-L)}$, which requires the computational complexity on the order of $\textsf{O}(M^3)$, where we have assumed $L= \textsf{O}(M)$, e.g., $L=\frac{M}{2}$; 2)~the calculation of pseudo-spectrum, which imposes the  complexity of $\textsf{O}(N L P)$ (note that this procedure involves only matrix multiplication and has highly efficient parallel acceleration, see \cite{roy1989esprit-estimation,Li2021}); 3)~the estimation of the unknown channel gains, which requires the computational complexity of $\textsf{O}(PM)$. Therefore, the overall complexity of our proposed estimator is on the order of $\textsf{O}(M^3)$.

For the MU-MIMO case with $K$ users, the computational complexity of our estimator is obviously on the order of $\textsf{O}(KM^3)$. In massive MIMO systems, $M$ is very large, and therefore this computational complexity is still too high. In Section~\ref{S5}, we further design a fast rank-1 subspace CSI estimator, relying on the randomized low-rank approximation technique \cite{woodruff2014sketching,wang2016spsd,drineas2005on,Li2019Fast}, with which the computational complexity of our channel estimator can be reduced dramatically, whilst the AoA estimation accuracy suffers no degradation.

%\subsection{Extension to Hybrid Massive MIMO}\label{S3.4}
%
%For the uplink channel estimation in the case of the hybrid massive MIMO architecture, the received signal reads:
%\begin{align} \label{eq:H_hybrid} % eq20
%\bm{Y} =& \bm{D} \bm{H} \bm{X}^{\rm H} + \bm{N},
%\end{align}
%where $\bm{D} \in \mathbb{C}^{M_B \times M}$ is the equivalent combining matrix composed of the radio-frequency (RF) combiners and the baseband (BB) combiners. Furthermore, we have $\vect(\bm{Y})=\big(\bm{X}^{\rm H} \otimes \bm{D}\big) \vect(\bm{H})+\vect(\bm{N)}$. Since we focus on the estimation of the unknown channel matrix $\bm{H}$, the combining matrix $\bm{D}$ is assumed to be known at this stage. For the $k$th user, we thus reformulate an equivalent model $\bm{Y}_k=\bm{h}_k\widetilde{\bm{x}}_k^{\rm H}+\bm{N}_k$, where $\widetilde{\bm{x}}_k=\text{vec}\big(\bm{x}_k \otimes \bm{D}\big)\in\mathbb{C}^{(BM_BM)\times 1}$ is the equivalent input vector, $\bm{Y}_k\in\mathbb{C}^{M\times (BM_BM)}$ is the receive signal matrix, and $\bm{N}_k\in\mathbb{C}^{M\times (BM_BM)}$ is the noise matrix. As such, our new CSI estimator is directly applicable.

\section{Theoretical Performance Analysis}\label{S4}

In this section, we analyze the accuracy of our rank-1 subspace channel estimator for massive MIMO uplink.

\vspace{-0.3cm}
\subsection{Error of Estimated Channel AoAs}\label{S4.1}
We first investigate the estimation accuracy of the AoA information of each user. We easily note that the norm of the Gaussian noise matrix in (\ref{eq:rank_limited}) obeys the following rule:
\begin{equation}\label{eq21}
\big\|\bar{\bm{N}}\big\|_2 \textcolor{black}{=} \textsf{O}\left(\sqrt{M\log_2(M)} \right).
\end{equation}
On this basis, if the stack length is chosen to be $L\thicksim \textsf{O}(M)$\textcolor{black}{\cite{liao2016music}}, e.g., $L=M/2$, we can acquire the accurate pseudo-spectrum estimation. When $M\rightarrow\infty$ and the Gaussian noise matrix $\bar{\bm{N}}$ meets the condition, $\big\|\bar{\bm{N}}\big\|_2^2\thicksim M\log_2(M)$, the error of the estimated pseudo-spectrum is bounded \cite{liao2016music}:
\begin{equation}\label{eq:error_spectrum} % eq22
|P(\theta)- P_{0}(\theta)| \textcolor{black}{=} \textsf{O}\left(\frac{\sqrt{\log_2(M)}}{\sqrt{M}} \right),
\end{equation}
where $P_{0}(\theta)$ is the exact spatial pseudo-spectrum in the absence of additive noise. As such, for massive MIMO uplink, we are able to obtain the highly accurate spatial pseudo-spectrum, even in the presence of additive noise.

Furthermore, the error of the estimated AoA is also bounded, as specified by the following Lemma~\ref{L1} \cite{liao2016music}.

\begin{lemma} \label{L1}
Assume that $L\geq P$ and $M-L\geq P$. Define
\begin{align}\label{eq23}
f(L) \triangleq & \frac{1}{L}\sqrt{\frac{2}{\pi}}\left( \frac{2}{\pi} - \frac{4}{L} \right)^{-\frac{1}{2}} .
\end{align}
Denote the minimal gap between two adjacent AoAs as $d=\min_{p\neq p'}\big|\theta_{p,k} -\theta_{p',k}\big|$. If
\begin{equation}\label{eq:minimal_gap} % eq24
d \geq \max \left\{ f(L),f(M-L)\right \},
\end{equation}
then the AoA estimation error is bounded for large $M$:
\begin{equation}\label{eq:AoA_error} % eq25
\delta_{p,k}\triangleq\big|\widehat{\theta}_{p,k}-\theta_{p,k}\big| \textcolor{black}{=} \textsf{O}\left( {\sqrt{\log_2(M)}}\big/ {M^{\frac{3}{2}}} \right) , ~ \forall p .
\end{equation}
\end{lemma}

Based on Lemma~\ref{L1}, it is seen that when $M$ is sufficiently large, as is the case for massive MIMO, the estimated AoA information is highly accurate. Also observe that the AoA estimation error meets the necessary condition for the unbiased channel gain estimate stated in Theorem~\ref{T2}.
{\color{black}More importantly, this high resolution AoA information enables the near-ML estimation of unknown CSI, as seen late in Section IV-C.}

%the standard MUSIC cannot be applied in this case of single snapshot signal ${\bf y}_k$, thus the comparative analysis between our method and MUSIC are not provided.
%Indeed, the randomized matrix decomposition technique in this work also can be used to reduce the complexity of traditional MUSIC algorithm, as seen in our previous works \cite{15,16}.

%\begin{figure}[!t]
%\centering
%\includegraphics[width=8.5cm]{figure/Figure2.pdf}
%\vspace{-4mm}
%\caption{\color{red} Estimated pseudo-spectrum for the AoA information estimation, where $M$=256, $P$=7, $L=128$, SNR=12\,dB and $N=?$.}
%\vspace{-4mm}
%\label{fig:1}
%\end{figure}

%Fig.~\ref{fig:1} shows an example of our estimated pseudo-spectrum. It can be shown that
%\begin{align*}
% \max\limits_{\theta} |P(\theta)- P_{0}(\theta)| =& ? \thicksim \textsf{O}\left( {\sqrt{\log_2(M)}}\big/ \sqrt{M} \right)  ~ \textsf{O}(0.2).
%\end{align*}
%This verifies (\ref{eq:error_spectrum}). Furthermore, we have
%\begin{align*}
% \max\limits_{0\le p\le 6} \big|\widehat{\theta}_{p}-\theta_{p}\big| =& ? \thicksim \textsf{O}\left( {\sqrt{\log_2(M)}}\big/ {M^{\frac{3}{2}}} \right) ~ \textsf{O}(0.001) .
%\end{align*}
%(in radian?) This agrees with Lemma~\ref{L1}.

\vspace{-0.3cm}
\subsection{Error of Estimated Channel Gains}\label{S4.2}
We now analyze the estimation accuracy of the channel gain estimates, obtained by the post-reception beamforming using our highly accurate AoA estimates. Denote the error of the $p$th AoA of the $k$th UE as $\delta_{p,k} = \big|\widehat{\theta}_{p,k}-\theta_{p,k}\big|$ and the variance of the pilot signal as $\sigma_x^2$. After the post-reception beamforming, the noise term $n_M$ in (\ref{eq:post_beam}) of Appendix~\ref{ApA} has the variance
\begin{align}\label{eq:def_hk} % eq26
\bar{\sigma}_M^2 =&  \frac{1}{M} \exp \left( - 2\ln(2) M^2 \delta_{p,k}^2 \right) \frac{\sigma_n^2}{B\sigma_x^2} .
\end{align}
Here, the approximate Gaussian beamformer is considered; see (\ref{eq:Gauss_beamformer}) in Appendix~\ref{ApB}. From (\ref{eq:AoA_error}), for large $M$, we have $\delta_{p,k} \rightarrow \textsf{O}\big(\sqrt{\log_2(M)}\big/M^{1+\epsilon}\big)$, with $\epsilon \rightarrow 1/2$. By applying the Taylor series expansion $\exp(x)=1+x+o(x^2)$, we have
\begin{align} \label{eq:def_hk1} % eq27
\bar{\sigma}_M^2 =& \frac{1}{M} \bigg(\! 1 -  2\ln(2) \textsf{O}\bigg( \frac{\log_2(M)}{M^{2\epsilon}} \bigg)\! \bigg) \frac{\sigma_n^2}{B\sigma_x^2} \overset{M~\text{is large}}{\longrightarrow} \frac{\sigma_n^2}{MB\sigma_x^2} .
\end{align}

Moreover, the channel gain estimation is the classical ML estimation (\ref{eq:ML_esti}), whose likelihood density is a complex Gaussian distribution with mean $\alpha_{p,k}$ and variance $\bar{\sigma}_M^2$, namely,
\begin{equation}\label{eq:Gauss_approximation} % eq28
p\left( \alpha_{p,k} - \frac{1}{M}\bm{a}_M^{\rm H}\big(\widehat{\theta}_{p,k}\big)\bm{y}_k\right) \sim \mathcal{CN}\big(\alpha_{p,k},\bar{\sigma}_M^2)\big.
\end{equation}
The following theorem shows that the estimates $\big\{\widehat{\alpha}_{p,k}\big\}_{p=0}^{P-1}$ are unbiased with the estimation MSE bound $\sigma_n^2\big/(MB\sigma_x^2)$.

\begin{theorem} \label{T3}
When $L = \textsf{O}(M)$ (e.g., $L=M/2$) and the minimal gap between AoAs satisfies (\ref{eq:minimal_gap}), the channel gain estimated via the rank-1 subspace method is asymptotically unbiased for a large $M$, namely,
\begin{align} \label{eq29}
\mathsf{E}\big[\widehat{\alpha}_{p,k}\big] =& \alpha_{p,k}, ~ \text{for large } M,
\end{align}
and the estimation MSE is bounded by $\sigma_n^2\big/(MB\sigma_x^2)$.
\end{theorem}

\begin{proof}
See Appendix~\ref{ApC}.
\end{proof}

Note that, unlike the classical FFT method whose channel gain estimate is biased, our rank-1 subspace method now attains the unbiased channel gain estimate.

\vspace{-0.3cm}
\subsection{Error of Estimated Channel Vector}\label{S4.3}
Based on the above error analysis on the estimates of the AoAs and the channel gains, in the angular domain, the residual error of the $k$th channel vector can be written as:
\begin{align}\label{eq:subset_s} % eq30
& \mathsf{E} \Big[\big\|\widehat{\bm{h}}_k - \bm{h}_k\big\|_2^2\Big] = \mathsf{E} \left[\sum\nolimits_{p=0}^{P-1} \left|\widehat{h}_{p,k}-h_{p,k}\right|^2 \right] \nonumber \\
& \hspace*{10mm} \overset{\delta_{p,k}\rightarrow 0}{\approx} \mathsf{E} \left[ \sum\nolimits_{p=0}^{P-1} \left|\widehat{\alpha}_{p,k} \cos(\delta_{p,k})-\alpha_{p,k}\right|^2 \right].
\end{align}

\textcolor{black}{
Recall from (\ref{eq:AoA_error}) that the MSE of AoA estimates of our new method is $\textsf{O}\left( \frac{\log_2 M} {M^3} \right)$, while that of the ML method is $\textsf{O}\left( \frac{1} {M^3} \right)$ \cite{stoica1989music}. We argue that, although the MSE gap between our method and ML is $\textsf{O}\left( \log_2M \right)$, the near-ML CSI estimation will be attained in the high SNR region, given the shared dominate decaying rate of $\textsf{O}\left( \frac{1} {M^3} \right)$.}

\textcolor{black}{
To be specific, when the number of antennas is sufficiently large, e.g., $M\geq 100$, we are able to omit the error of AoA estimates in the high SNR region.
That is to say, we will asymptotically have $\cos\left(\delta_{p,k}\right) \overset{M~\text{is large}}{\longrightarrow} 1$.}
Further consider that the channel gains were estimated via the ML criterion \cite{kay1993fundamentals}, the CRLB of the channel gain estimation is achievable. As a result, the estimation error of the $p$th component in the $k$th channel vector is:
\begin{align}\label{eq31}
\mathsf{E}\left[\big|\widehat{h}_{p,k} - h_{p,k}\big|^2\right]  \overset{M~\text{is large}} \approx & \mathsf{E} \left[\big|\widehat{\alpha}_{p,k} -\alpha_{p,k}\big|^2 \right] \geq  \frac{\sigma_n^2}{MB\sigma_x^2}.
\end{align}

\vspace{-0.5cm}
\subsection{Gain Over Linear MMSE Estimator}\label{S4.4}
Based on the signal model (\ref{eq:linear_gaussian}), the classical linear MMSE estimator would also obtain an accurate channel estimation. But it requires the MIMO channel covariance matrix $\bm{R}_{\bm H} \in \mathbb{C}^{MK\times MK}$ and the channel AWGN covariance matrix $\bm{R}_{\rm N}\in \mathbb{C}^{KB\times KB}$ to compute the channel estimate \cite{kay1993fundamentals}
\begin{align}\label{eq32}
\text{vec}\left(\widehat{\bm{H}}^{\rm mmse}\right) =& \bm{R}_{\bm{H}} \tilde{\bm{X}}^{\rm H} \big(\tilde{\bm{X}} \bm{R}_{\bm{H}} \tilde{\bm{X}}^{\rm H} + \bm{R}_{\bm{N}}\big)^{-1} \text{vec}\left(\bm{Y}\right),
\end{align}
where $\tilde{\bm{X}}=\bm{X}^T\bigotimes\bm{I}$, and $\bigotimes$ is the Kronecker product. Note that, the acquisition of the MIMO channel covariance matrix and its inverse are very costly. So, another simpler LS channel estimator computes the MIMO channel matrix estimate as
\begin{align}\label{eq33}
\widehat{\bm{H}}^{\rm ls} =& \bm{Y} \bm{X}^\dag .
\end{align}
Yet it has a larger estimation error than the linear MMSE estimator.

Note that the linear MMSE estimator should not outperform the ML method whose CRLB of the estimation error is given in \cite{dong2002optimal,berriche2004cramer}.
Therefore, using this CRLB, we have
\begin{align}\label{eq34}
\mathsf{E}\left[\Big\|\widehat{h}_{p,k}^{\rm mmse} - h_{p,k}\Big\|_2^2\right] \geq \frac{\sigma_n^2}{B\sigma_x^2 + \rho_h^2 \sigma_n^2} ,
\end{align}
where $\mathsf{E}\left[\frac{\partial \ln p_{\bm{h}}(\bm{h})}{\partial \bm{h}^*} \frac{\partial^{\rm H} \ln p_{\bm{h}}(\bm{h})}{\partial \bm{h}^*} \right]\triangleq\rho_h^2 \bm{I}_M$; $\rho_h^2=P+1$ and $p_{\bm{h}}(\bm{h})$ is the prior distribution of the channel vector $\bm{h}_k$. %Usually, we normalize $\sigma_x^2$ to be 1.

Comparing (\ref{eq31}) with (\ref{eq34}), it can be seen that our rank-1 subspace estimator is capable of attaining an higher estimation accuracy than the linear MMSE estimator.

\begin{theorem}\label{T4}
Define the normalized MSE of the channel estimation as $\gamma=\mathsf{E}\Big[\big\|\widehat{\bm{H}}-\bm{H}\big\|_F^2\big/\|\bm{H}\|_F^2\Big]$. In order to achieve a target normalized MSE $\gamma$ (e.g., $\gamma<10^{-2}$), the required SNR for our rank-1 subspace method is ${\rm snr}_1\triangleq 10\log_{10} (E_1 /\sigma_n^2)$, while the required SNR for the linear MMSE estimator is ${\rm snr}_2\triangleq 10\log_{10}(E_{2}/\sigma_n^2)$, where $E_1$ and $E_2$ are the respective transmission powers of the two schemes. When $M$ is large and the stack length is $L= \textsf{O}(M)$, we have:
\begin{align}\label{eq35}
{\rm snr}_{2} - {\rm snr}_{1} \textcolor{black}{=} \textsf{O}\big(\log_{10}(M)\big).
\end{align}
\end{theorem}

The proof of Theorem~\ref{T4} is straightforward.
\textcolor{black}{From (31), we have ${\rm snr}_1\triangleq 10\log_{10} (E_1 /\sigma_n^2) =  10\log_{10}(\frac{\sigma_n^2}{MB\sigma_x^2}) $.
While for the linear MMSE method, we have ${\rm snr}_2\triangleq 10\log_{10} (E_2 /\sigma_n^2) =  10\log_{10} (\frac{\sigma_n^2}{B\sigma_x^2 + \rho_h^2 \sigma_n^2})$.
Thus, we obtain ${\rm snr}_{2} - {\rm snr}_{1} =  10\log_{10} \left[\frac{M}{1 + \rho_h^2/B \times \sigma_n^2/ \sigma_x^2}\right]  = \textsf{O}\big(\log_{10}(M)\big)$, when $\sigma_n^2/\sigma_x^2$ is small.}

\begin{remark}\label{R1}
Theoretical analysis shows that our new estimator is capable of acquiring the high-resolution AoA for large $M$ and $L\thicksim \textsf{O}(M)$. When $M$ is small, this result may be invalid. Later in the numerical simulation, it will be shown that the above theoretical result holds when $M>M_0$ with $M_0=64$. Although our rank-1 subspace estimator achieves excellent estimation performance for massive MIMO systems, its computational complexity on the order of $\textsf{O}(KM^3)$ still makes it impractical. Note that the same problem is also applied to the linear  MMSE estimator\footnote{For the correlated MIMO channel matrix, the computational complexity of the linear MMSE estimator becomes $\textbf{O}\big((MB)^3\big)$ \cite{shariati2014low,Li2020Randomized}.}, which has the complexity of $\textsf{O}(M^3)$. It is critical to achieve a low-complexity deployment for massive MIMO system, while maintaining the high performance.
\end{remark}

\section{Fast Rank-1 Subspace Estimator}\label{S5}
In this section, we design a fast implementation of our rank-1 subspace method.
From Subsection~\ref{S3.3}, the complexity of this rank-1 subspace estimator mainly comes from the SVD of the Hankel matrix $\bar{\bm{Y}}$. Without loss of generality, in the sequel we assume $L=M/2$. Therefore, from (\ref{eq:rank_limited}), the space-embedding Hankel matrix $\bar{\bm{Y}}\in \mathbb{C}^{L\times L}$ becomes a symmetric positive semi-definite (SPSD) matrix, i.e., $\bar{\bm{Y}}=\bar{\bm{Y}}^{\rm T}$. As discussed previously, the rank of this SPSD matrix $\bar{\bm{Y}}$ is restricted. By exploiting this low-rank property, we resort to a randomized low-rank approximation technique that is widely applied in linear algebra and scientific computing \cite{woodruff2014sketching,wang2016spsd,drineas2005on,Li2020Randomized,2021Random,2023Machine}, to construct a fast rank-1 subspace CSI estimator without impeding the estimation accuracy.
It is noteworthy that, although this fast method may substantially reduce the computational complexity in the case of low-rank channel matrix, it is not an essential component of our two-stage CSI estimator.
For example, for the full-rank channel matrix, there is no need to apply it.

\vspace{-0.3cm}
\subsection{Randomized Low-rank Approximation}\label{S5.1}
In order to minimize the overall computational complexity, we focus on \emph{approximate} rather than \emph{exact} computation of the SVD of $\bar{\bm{Y}}\in \mathbb{C}^{L \times L}$. To achieve this goal, we first approximate $\bar{\bm{Y}}$ with a special structure \cite{Li2019Fast}:
\begin{align}\label{eq:Y_appro} % eq36
\bar{\bm{Y}} \approx &  \bm{C} \bm{W} \bm{C}^{\rm H} ,
\end{align}
where the small matrix sketch $\bm{C}\in \mathbb{C}^{L \times s}$ is abstracted from $\bar{\bm{Y}}$ via random sampling\footnote{Note that, in the general case ($L\neq M/2$), the Hankel matrix $\bar{\bm{Y}}$, constructed from the received signal vector $\y$, is not symmetric. For simplicity, here we assume $L = M/2$.}, while the small matrix $\bm{W}\in \mathbb{C}^{s \times s}$ is computed based on the chosen sketch $\bm{C}$.

To abstract $\bm{C}$ from $\bar{\bm{Y}}$ via random sampling, we specify the sampling length $s$ that satisfy $\text{rank}\big(\bar{\bm{Y}}\big)\leq s\ll M$, and randomly pick up the $s$ columns of $\bar{\bm{Y}}$ according to the uniform distribution $\mathcal{U}\{1,2, \cdots ,L\}$. Denoting the sampled column indexes as the set $\mathcal{I}\triangleq\{i_1,i_2,\cdots, i_s\}$ with $|\mathcal{I}|=s$, we have
\begin{align}\label{eq37}
\bm{C} =& \bar{\bm{Y}}(:,\mathcal{I}) = \bar{\bm{Y}} \bm{\Pi} ,
\end{align}
where $\bm{\Pi}\in \mathbb{R}^{L \times s}$ is the uniform sampling matrix, which selects the $s$ columns of $\frac{\sqrt{L}}{\sqrt{s}}\bm{I}_L$ with the column index set $\mathcal{I}$, namely, $\bm{\Pi}=\frac{\sqrt{L}}{\sqrt{s}}\bm{I}_L(:,\mathcal{I})$.

Given the small sketch $\bm{C}\in \mathbb{C}^{L \times s}$, the weight matrix $\bm{W}\in \mathbb{C}^{s \times s}$ is determined, by minimizing the approximation error
\begin{align}\label{eq38}
\bm{W} =& \arg \min_{\bm{W}' \in \mathbb{C}^{s \times s}} \big\|\bar{\bm{Y}} - \bm{C} \bm{W}' \bm{C}^{\rm H}\big\|_F^2 = \bm{C}^\dag \bar{\bm{Y}} \big(\bm{C}^\dag\big)^{\rm H}.
\end{align}
Considering that the above exact solution requires the computational complexity of $\textsf{O}\big(2 s^2 M + s M^2\big)$, we choose to solve it approximately. Specifically, since the optimization problem (\ref{eq38}) is \emph{over-determined}, $\bm{W}$ can be derived in a more efficient way, by further sketching both $\bar{\bm{Y}}$ and $\bm{C} \bm{W} \bm{C}^{\rm H}$, as is done by the Nystr$\ddot{\text{o}}$m method \cite{drineas2005on,Li2019Fast,gittens2016revisiting}. Hence, we have
\begin{align}\label{eq39}
\bm{W} =& \arg \min_{\bm{W}' \in \mathbb{C}^{s \times s}} \big\|\bm{\Pi}^{\rm T}\big(\bar{\bm{Y}} - \bm{C} \bm{W}' \bm{C}^{\rm H}\big)\bm{\Pi} \big\|_F^2, \nonumber \\
 =& \big(\bm{\Pi}^{\rm T}\bar{\bm{Y}} \bm{\Pi}\big)^\dag = \big(\bar{\bm{Y}}(\mathcal{I},\mathcal{I})\big)^\dag.
\end{align}

With $\bm{C}$ and $\bm{W}$, we obtain the randomized low-rank approximation $\bar{\bm{Y}}$ as in (\ref{eq:Y_appro}). To compute its SVD, we first obtain the SVD of $\bm{C}$ as $\bm{C}=\bm{U}_c\bm{\Sigma}_c\bm{V}_c^{\rm H}$, and then approximate $\bar{\bm{Y}}$ by:
\begin{align}\label{eq40}
\bar{\bm{Y}} \approx \bm{U}_c \underbrace{\bm{\Sigma}_c \bm{V}_c^{\rm H} \bm{W} \bm{V}_c \bm{\Sigma}_c^{\rm H}}_{\triangleq\bm{B}} \bm{U}_c^{\rm H} = \bm{U}_c\bm{B} \bm{U}_c^{\rm H}.
\end{align}
As $\bm{B}$ is a Hermitian matrix, i.e., $\bm{B}=\bm{B}^{\rm H}$, it has the SVD of $\bm{B}=\bm{U}_B\bm{\Sigma}_B\bm{U}_B^{\rm H}$. Substituting the SVD of $\bm{B}$ into (\ref{eq40}) leads to the approximate rank-restricted SVD of $\bar{\bm{Y}}$ as
\begin{align}\label{eq41}
\bar{\bm{Y}} \approx & \bm{U}_c \bm{U}_B \bm{\Sigma}_B \bm{U}_B^{\rm T} \bm{U}_c^{\rm H} = \widetilde{\bm{U}}_P \bm{\Sigma}_B \widetilde{\bm{U}}_P^{\rm H} ,
\end{align}
where the unitary matrix $\widetilde{\bm{U}}_P \in \mathbb{C}^{L \times s}$, given by
\begin{equation}\label{eq:subspace_appro} % eq42
\widetilde{\bm{U}}_P\triangleq\bm{U}_c\bm{U}_B ,
\end{equation}
defines approximately the signal subspace. On this basis, we can approximate the pseudo-spectrum (\ref{eq12}) by
\begin{align}\label{eq:spectrum_appro} % eq43
\widetilde{P}(\theta) =& \frac{1}{\bm{a}_L^{\rm H}(\theta)\big(\bm{I}_{L}-\widetilde{\bm{U}}_{P}\widetilde{\bm{U}}_{P}^{\rm H}\big)\bm{a}_L(\theta)} ,
\end{align}
and search for its $P$ highest peaks on the $N$ gridded angles $\{-\pi+n \Delta\theta,~n=0,1,\cdots,N-1\}$.
In this stage, other subspace-based methods, e.g., ESPRIT, can also be applied to compute the above pseudo-spectrum and thus acquire unknown AoAs.
When the exact number of $P$ is known as {\it a priori}, highly accurate AoAs can be attained with ESPRIT.
While this prior information is not known, the above MUSIC method can still acquire the accurate result.

\vspace{-0.3cm}
\subsection{Complexity Analysis}\label{S5.2}
The computational complexity of our fast rank-1 subspace estimator for the single UE involves three parts: 1)~compute the SVD of $\bar{\bm{Y}}$ approximately, which has the complexity of $\textsf{O}\big(s^2M+s^3\big)$, 2)~evaluate the spatial pseudo-spectrum, which has the complexity of $\textsf{O}(NLP)$, and 3)~estimate the channel gains, with the complexity of $\textsf{O}(M)$. Since both the number of users $K$ and the number of paths $P$ are very small ($\ll M$) and further considering that $L=\frac{M}{2}$ and in practice the number of sampled columns $s=\textsf{O}(P)$, the overall computational complexity of our fast method for the $K$-user case is
\begin{align}\label{eq44}
C_{\text{fast rank-1}} \textcolor{black}{=} \textsf{O}\big( K (P^2+NP) M\big) .
\end{align}
Hence, the complexity of our fast method scales linearly with $M$. This is two orders of magnitude lower than that of the linear MMSE method, which has the complexity of $\textsf{O}(M^3)$.
\textcolor{black}{This is because in practice we usually have $P<6, M>64, K\ll M$ in the multi-users massive MIMO scenario \cite{akdeniz2014millimeter}, and hence $ K (P^2+NP) M \ll M^3$. }

\begin{table}[t!]
\footnotesize
\vspace*{-1mm}
\caption{Computational Complexity of Various CSI Estimators.}
\label{tab1}
\vspace{-4mm}
\begin{center}
\begin{tabular}{ll}
\hline
\multicolumn{1}{c}{Estimator} & \multicolumn{1}{c}{Computational complexity} \\ \hline
LS\cite{marzetta2010noncooperative} & $\textsf{O}\big(K B M + K B^2\big)$ \\
SVP\cite{shen2015joint}$^1$ & $\textsf{O}\big(t_{svp}r_{est}(K^2 +K) M +2K^2\big)$ \\
FFT\cite{fan2017angle} & $\textsf{O}\big( K( B + \log_2(M))M + K P^2\big)$ \\
WPEACH\cite{shariati2014low}$^2$ & $\textsf{O}\big( N_l( B K^2 + B^2 K)M^3 + N_l^3 \big)$ \\
Linear MMSE\cite{shariati2014robust}$^3$ & $\textsf{O}\big( B( 2 K^2 + 2B K + B^2 )M^3 \big)$           \\
LSE-SMP\cite{huang2019iterative}$^4$ & $\textsf{O}\big( t_{lse-smp} K (B^2 + K^2) M^3 \big)$ \\
Random-MMSE\cite{Li2020Randomized}$^5$ & $\textsf{O}\big( B_1^3 M^3 + s_r^3 B_2^2 + s_r^2s_c \big)$   \\
Proposed$^6$ & $\textsf{O}\big(K M^3 + K N P M\big)$ \\
Proposed fast$^6$ & $\textsf{O}\big( K (P^2+NP) M\big)$ \\ \hline
\end{tabular}
\end{center}
\vspace*{-2mm}
\begin{tablenotes}
\footnotesize
\item[1] 1. $t_{svp}$: the number of iterations, and $r_{est}$: the order of low rank.
\item[2] 2. $N_l$: truncation order of Taylor series expansion.
\item[3] 3. $\textsf{O}\big(M^3 \big)$ for uncorrelated MIMO channels.
\item[4] 4. $t_{lse-smp}$: the number of iterations.
\item[5] 5. $B_1/B_2$: the pilot lengths, and $s_r / s_c$: the sampling lengths.
\item[6] 6. Complexity $\textsf{O}(K N P M)$ for matrix multiplications has fast parallel acceleration implementation \cite{Li2021}.
\end{tablenotes}
\vspace*{-5mm}
\end{table}

Table~\ref{tab1} compares the computational complexity of the following massive MIMO channel estimators: the linear  LS \cite{marzetta2010noncooperative}, the SVP \cite{shen2015joint}, the FFT \cite{fan2017angle}, the WPEACH \cite{shariati2014low}, the linear MMSE \cite{shariati2014robust}, the LSE-SMP \cite{huang2019iterative}, the random-MMSE \cite{Li2020Randomized}, our rank-1 subspace estimator, and our fast rank-1 subspace estimator. We will show that our low-complexity fast rank-1 subspace estimator is capable of outperforming the linear MMSE estimator through both theoretical analysis and numerical simulation.

\vspace{-0.4cm}
\subsection{Theoretical Analysis on Approximated Pseudo-spectrum}\label{S5.3}
We first give the error bound on the approximate pseudo-spectrum $\widetilde{P}(\theta)$ (\ref{eq:spectrum_appro}) computed based on the subspace $\widetilde{\bm{U}}_P$.
To be self-contained, here we provide the following Theorem 5 to demonstrate the accuracy of the fast method, while the proof of Theorem 5 can be found in our previous work \cite{Li2019Fast}.
	
\begin{theorem} \label{T5}
Give $\bm{C}=\bar{\bm{Y}}\bm{\Pi}$ and $\bm{W}=\big(\bar{\bm{Y}}(\mathcal{I},\mathcal{I})\big)^\dag$, where $\bm{\Pi}\in \mathbb{R}^{L\times s}$ is a uniform sampling matrix. Let the signal subspace $\widetilde{\bm{U}}_P$ be computed by (\ref{eq:subspace_appro}) and the pseudo-spectrum be approximated by (\ref{eq:spectrum_appro}). Define the coherence of $\widetilde{\bm{U}}_P$ as $\mu\big(\widetilde{\bm{U}}_P\big)\triangleq \frac{L}{P}\max\limits_{1\le m\le L} \big\|\widetilde{\bm{U}}_P(m,:)\big\|_2^2 \in \left[1, ~ \frac{L}{P}\right]$ and give a constant $\delta \in (0,~ 1)$. When the sampling length $s$ satisfies
\begin{align}\label{eq:sampling_len} % eq45
s \geq & 4.5 P \log_2 \left(\tfrac{P}{\delta} \right) \mu\big(\widetilde{\bm{U}}_P\big),
\end{align}
the following relation holds with probability at least $1-\delta$
\begin{align}\label{eq:lower_bound} % eq46
\sqrt{P(\theta)} \leq & \left( 1 + \tfrac{L}{\sqrt{s}} \tfrac{\sigma_{P+1}\big(\bar{\bm{Y}}\big)}{\sigma_{P}\big(\bar{\bm{Y}}\big)} \right)  \sqrt{\widetilde{P}(\theta)} .
\end{align}
\end{theorem}

Theorem~\ref{T5} gives the lower error bound on the approximate pseudo-spectrum.
When the SNR is high, the spectrum gap $\tfrac{\sigma_{P+1}\big(\bar{\bm{Y}})}{\sigma_{P}\big(\bar{\bm{Y}}\big)}$ is very small, i.e., $\tfrac{\sigma_{P+1}\big(\bar{\bm{Y}})}{\sigma_{P}\big(\bar{\bm{Y}}\big)}\rightarrow 0$, and the approximation error becomes ignorable.

\begin{remark}\label{R2}
Based on the above lower bound, we can show that the approximated pseudo-spectrum should not miss the true AoAs. On one hand, the exact pseudo-spectrum peaks at a true AoA and is significantly larger than the noise baseline at the non-AoA region. Specifically, $P(\theta_p)/P(\theta_{null})\thicksim L$, where $\theta_{null}\in [0, ~ \pi]/\{\theta_p\}_{p=1}^P$. If a single AoA was missed in the approximated pseudo-spectrum, we expect to have $\widetilde{P}(\theta_p)\thicksim P(\theta_{null})$. On the other hand, Theorem~\ref{T5} suggests that $P(\theta_p)/\widetilde{P}(\theta_p) \leq (1+\epsilon)^2$, where $\epsilon\triangleq \tfrac{L}{\sqrt{s}} \tfrac{\sigma_{P+1}\big(\bar{\bm{Y}}\big)}{\sigma_{P}\big(\bar{\bm{Y}}\big)}$. Due to the low-rank property of Hankel matrix $\bar{\bm{Y}}$, in the high SNR region, we should have $\epsilon < 1$. Combining the both arguments results in the contradictory result of $(1+\epsilon)^2 \geq L$. This indicates that $\widetilde{P}(\theta_p)\thicksim P(\theta_{null})$ should be untrue, in the case of large $M$. Thus, the approximated pseudo-spectrum should accurately capture the real AoAs.
\end{remark}

Hence, like the exact pseudo-spectrum (\ref{eq12}), the approximated pseudo-spectrum (\ref{eq:spectrum_appro}) can be also used to acquire the accurately AoA information. Furthermore, the rest of the fast rank-1 subspace estimator are identical to those of the rank-1 subspace estimator. Therefore, we may conclude that the fast rank-1 subspace estimator should not degrade the achievable performance of the rank-1 subspace estimator.

%\vspace{-0.3cm}
%\subsection{Practical Implementation Note}\label{S5.4}
%In practice, even though the Hankel matrix $\bar{\bm{Y}}$ is low ranked, its singular values may decay slowly, due to the additive noise, particularly in the low SNR region. This may slightly affect the accuracy of the randomized low-rank approximation for $\bar{\bm{Y}}$.
%To combat this problem, we first compute the covariance matrix of this Hankel matrix $\bm{Y}_s=\bar{\bm{Y}}\bar{\bm{Y}}^{\rm H}$. Using $\bm{Y}_s$, the signal subspace of $\bar{\bm{Y}}$ can be computed equivalently. Since $\bar{\bm{Y}}$ is symmetric for $L=M/2$, based on (\ref{eq11}) we have
%\begin{align}\label{eq47}
%\bm{Y}_s=&\bm{U} \bm{\Sigma} \bm{U}^{\rm H} \bm{U} \bm{\Sigma} \bm{U}^{\rm H} = \bm{U} \bm{\Sigma}^2 \bm{U}^{\rm H} .
%\end{align}
%The singular values of $\bm{Y}_s$ decay more quickly, and we can compute the approximated signal subspace $\widetilde{\bm{U}}_P$ of $\bar{\bm{Y}}$ more accurately.

\section{Numerical Analysis}\label{S6}
In this section, we compare the performance and complexity of various massive MIMO channel estimators based on numerical simulations using several evaluation metrics, including the normalized mean square error (NMSE), the averaged CPU runtime, and the channel capacity.

\begin{figure}[!t]
\vspace{-4mm}
\centering
\includegraphics[width=6.2cm]{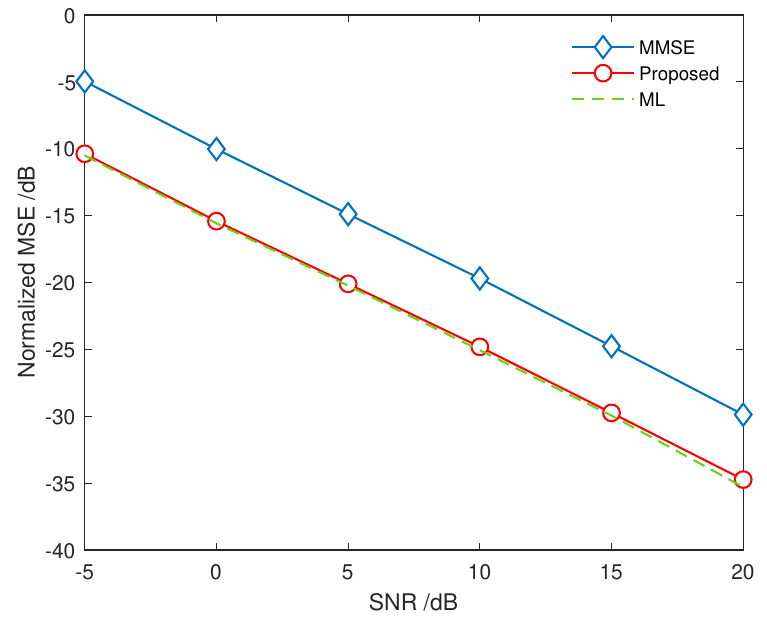}
\vspace{-4mm}
\caption{Normalized MSE performance of our rank-1 subspace estimator and the linear MMSE estimator as functions of SNR for the full-rank, sparse and uncorrelated channel, given $M=128$, $K=40$, $B=2K$ and $L=M/2$.}
\label{fig:1}
\vspace{-1mm}
\end{figure}

\vspace{-0.4cm}
\subsection{Full-rank and Sparse Channel}\label{S6.1}
We first compare our rank-1 subspace estimator with the classical linear MMSE method, in the case of full-rank and sparse channel matrix, i.e., $\text{rank}(\bm{H})=\min\{M,K\}=K$ and $P\ll M$. That means, the channels are uncorrelated.
In the simulation, we first randomly generate $P$ independent AoAs of each channel vector and $P$ i.i.d. Rayleigh fading gains.
With the assumption of ULA, these spatial channel parameters can be transformed to the temporal domain, and then the channel matrix is obtained.
As discussed, the linear MMSE method requires the channel covariance $\bm{R}_{\bm{H}}$. We assume that such {\it a priori} statistical information is available, thus leading to the \emph{genie-aided} linear MMSE estimator \cite{balevi2020massive}. This may be impractical as acquiring an accurate estimation of the channel covariance requires considerable time resources. By contrast, our method excludes such an impractical prior knowledge.

\begin{figure}[!t]
%\vspace*{-1mm}
\centering
\includegraphics[width=6cm]{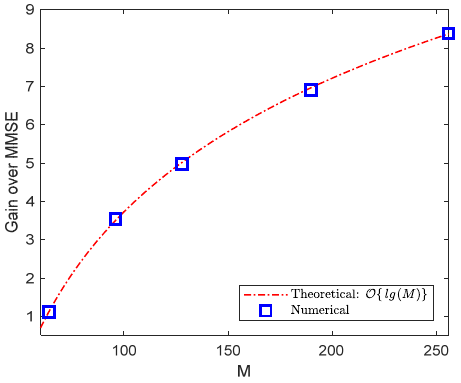}
\vspace{-4mm}
\caption{SNR gain of our rank-1 subspace method over the classical linear MMSE estimator as function of $M$ for the full-rank sparse channel of Fig.~\ref{fig:1}, with $K=40$, $B=2K$ and $L=M/2$.}
\vspace{-5mm}
\label{fig:2}
\end{figure}

Fig.~\ref{fig:1} compares the NMSE performance of our rank-1 subspace estimator with that of the classical linear MMSE estimator and the ML estimator, for the massive MIMO system with the number of antennas $M=128$, the number of users $K=40$ and the length of pilot signal $B=2K$. The number of paths is $P=5\thicksim 7$ for each channel vector, and we set the stack length to $L=M/2$ for our estimator.
Observe from Fig.~\ref{fig:1} that our new method indeed approximates the ML estimator, and dramatically outperforms the linear MMSE estimator. For example, at $\text{SNR}=20$\,dB, the NMSE of our rank-1 subspace estimator is more than 6\,dB lower than that of the linear MMSE method. This significant gain is mainly attributed to two key factors: 1)~the highly accurate AoA information acquired by the rank-1 subspace method, and 2)~the largely improved SNR by the \emph{post-reception} beamforming.
%{\color{blue}Note that, another major advantage of our new method is that, by providing the highly accurate initial soltion, it can dramatically reduce the computational complexity of the ML estimator in searching a massive space. For example, with this angular domain initialization, the likelihood function will be maximized after dozens of iterations when using a simplex downhill method.}

Fig.~\ref{fig:2} depicts the SNR gain achieved by our new method over the linear MMSE estimator under different $M$. First, this SNR gain increases with $M$, since the SNR after the post-reception beamforming increases with $M$. Second, the numerical results show that the attained gain over the linear MMSE method indeed shows a {\it logarithmic} increasing trend with $M$, which validates the theoretical result of Theorem~\ref{T4}. Therefore,
by applying our new estimator, the UEs far away from the BS can be admitted to establish reliable communication links, and the coverage area of massive MIMO system can be effectively extended.

\begin{figure}[!t]
\centering
\includegraphics[width=60mm]{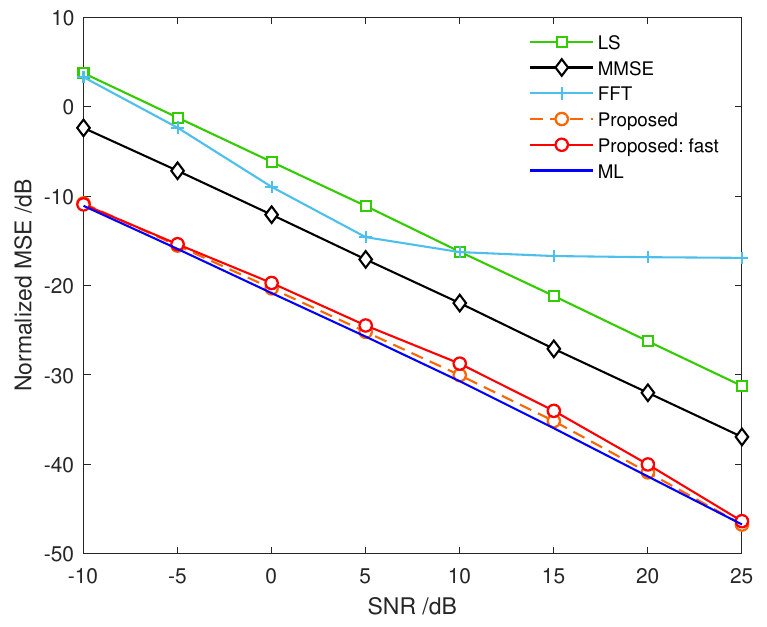}
\vspace{-4mm}
\caption{Normalized MSE performance of five channel estimators as functions of SNR for the full-rank, sparse and uncorrelated channel, given $M=256$, $K=40$, $B=2K$ and $L=M/2$.}
\label{fig:3}
\vspace{-1mm}
\end{figure}

Next we use the same full-rank channel model with $M=256$, $L=128$, $K=40$, $B=2K$ and $P=7$, to compare the NMSE performance of several state-of-the-art channel estimators in Fig.~\ref{fig:3}. Again, the propagation channels of $K$ UEs are uncorrelated. In this case, many low-rank based channel estimators, including the SVP and WPEACH methods, are inapplicable. As expected, the LS estimator and the FFT-based estimator acquire less accurate CSI than the linear MMSE method. As a classical angular-domain estimator \cite{fan2017angle}, the FFT-based method also acquire AoAs first. However, due to its limited resolution and biased channel gain estimation, it exhibits an error floor at the high SNR region, as can be clearly seen in Fig.~\ref{fig:3}. Observe that our fast rank-1 subspace estimator achieves the same NMSE performance as the rank-1 subspace estimator. The numerical result of Fig.~\ref{fig:3} thus validates the theoretical analysis of Subsection~\ref{S5.3}, namely, the approximation of the fast rank-1 subspace implementation does not degrade the achievable performance. Hence, the fast rank-1 subspace estimator significantly outperforms the linear MMSE estimator, in terms of channel estimation accuracy. Note that for the fast rank-1 subspace method, we set the sampling length to $s=\lceil 1.5 \times P\rceil$, i.e., $s \thicksim \textsf{O}(P)$, which is less strict than the requirement of (\ref{eq:sampling_len}); recall that (\ref{eq:sampling_len}) is a sufficient condition. With this setting on $s$, it is sufficient to attain a highly accurate AoA information.

\begin{figure}[!t]
	\begin{center}
		\subfigure[ ]{
			\label{fig:4a}
			\includegraphics[width=6.1cm]{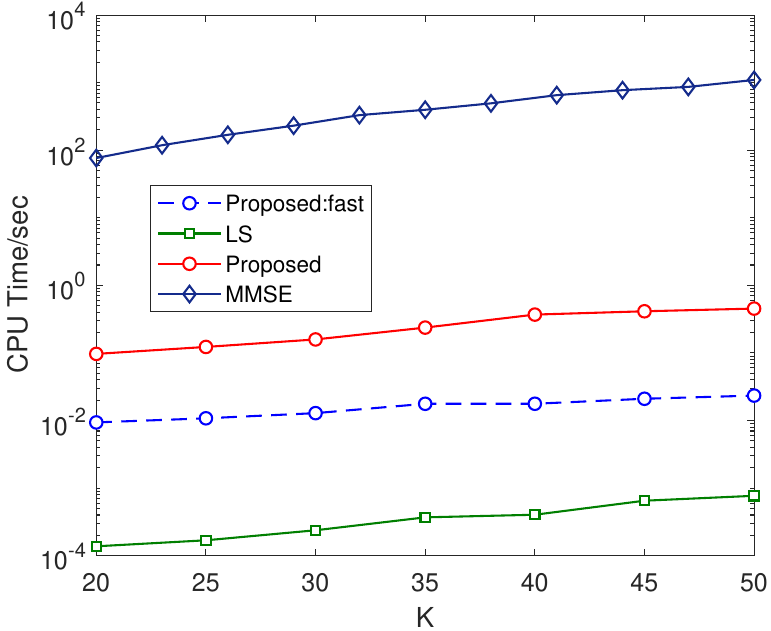}}
		\subfigure[ ]{
			\label{fig:4b}
			\includegraphics[width=6.1cm]{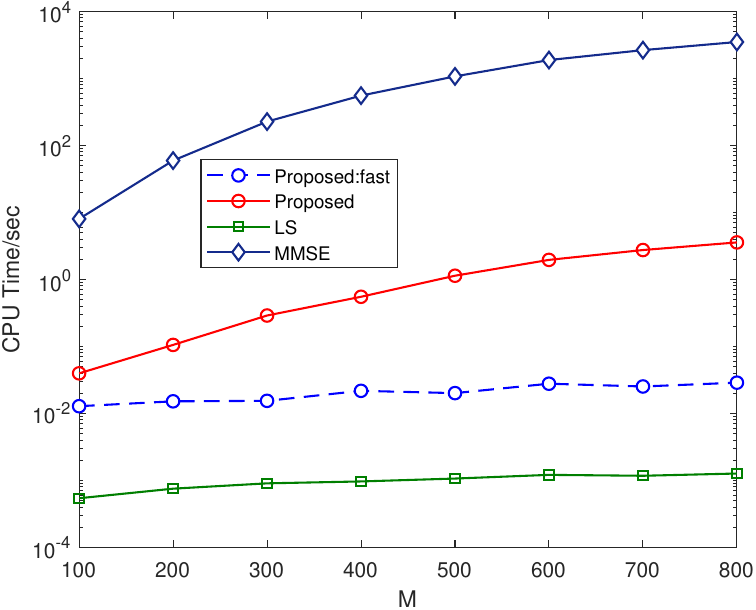}}
	\end{center}
	\vspace*{-5mm}
	\caption{Time complexity of four channel estimators for the full-rank sparse channel of Fig.~\ref{fig:3}: (a)~as functions of $K$, given $M=256$, $B=2K$ and $L=M/2$, and (b)~as functions of $M$, given $K=40$, $B=2K$ and $L=M/2$.}
	\label{fig:4}
	\vspace*{-1mm}
\end{figure}
To evaluate the computational complexity of various channel estimators, {\color{black}except for the comparative results in Table 1, we also} provide the CPU run-time to measure time complexity or processing delay. The operating system runs at 2.7\,GHz basic frequency with RAM 32\,GB. For the same full-rank system of Fig.~\ref{fig:3}, Fig.~\ref{fig:4a} compares the time complexity of four channel estimators by varying the number of users $K$. As expected, the LS estimator imposes the lowest computational complexity, which is its major advantage in practical applications. Observe that our fast rank-1 subspace method incurs a lower complexity compared to the linear MMSE estimator. This is remarkable, considering the fact that it significantly outperforms the linear MMSE method, in terms of channel estimation accuracy. Fig.~\ref{fig:4b} shows the time complexity of the four channel estimators under different $M$. As shown, the complexity of our fast method increases linearly with $M$. This verifies the complexity analysis of Subsection~\ref{S4.2}. It can also be seen that for $M>200$, the complexity of our fast estimator is dramatically lower than the linear MMSE estimator. In particular, for $M=800$, its complexity is around $10^6$ times lower than the linear MMSE method. This agrees with the analysis of Subsection~\ref{S4.2}, which indicates that for massive MIMO with large $M$, the complexity of our fast method is on two orders of magnitude lower than that of the linear MMSE method.

\begin{figure}[!t]
\centering
\includegraphics[width=60mm]{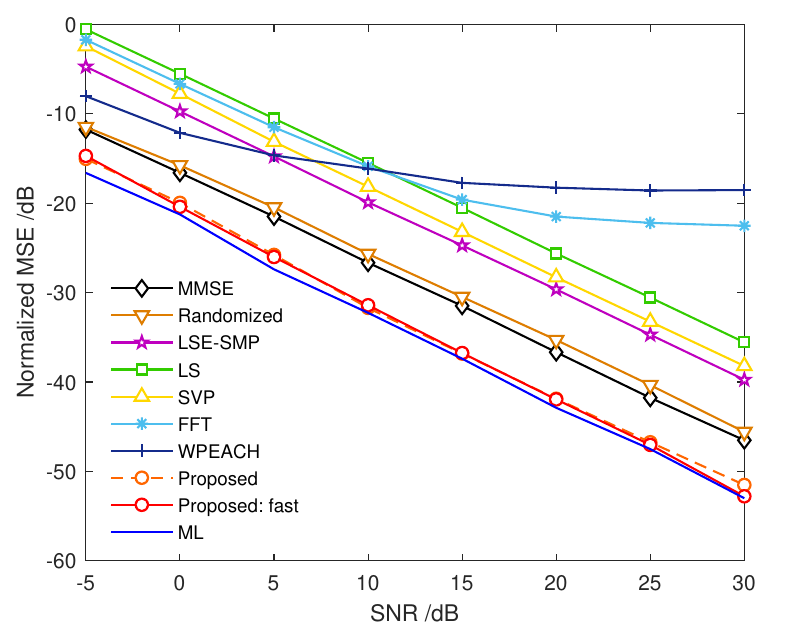}
\vspace{-4mm}
\caption{Normalized MSE performance of various channel estimator for the low-rank, sparse and correlated channel, given $\text{rank}(\bm{H})=20$, $M=256$, $K=40$, $L=M/2$, and $B=K$.}
\vspace*{-1mm}
\label{fig:5}
\end{figure}

\vspace{-0.4cm}
\subsection{Low-rank and Sparse Channel}\label{S6.2}
We further consider the case of low-rank and sparse massive MIMO, i.e., $\text{rank}(\bm{H})<K$ and $P\ll M$.
As reported in \cite{Xie2017An}, multiple users may experience the similar propagation environments, and hence the channel involves correlated paths. In this case, the channel matrix is characterized by both low rank and sparsity. Many channel estimators were developed to exploit these two features \cite{shen2015joint,fan2017angle}.
In the simulation, we firstly generate a set of random AoAs and fading gains, while assuming multiple channel vectors of clustered users share the same set of spatial parameters.
After the spatial channel parameters are transformed to the temporal domain, the low-rank and sparse channel matrix can be obtained.
We assume $P = 5\thicksim7$, $M=256$, $K=40$,  $\text{rank}(\textbf{H})=20$, $L=M/2$, and $B=2K$. For the fast method, the sampling length is $s=20$.

Fig. \ref{fig:5} compares the normalized MSE performance of various channel estimators. Similar to the FFT-based method, the WPEACH also exhibits a high error floor at the high SNR region. This is due to the approximation based on a limited expansion order~(e.g., $3\thicksim5$). The SVP achieves 3dB gain over the LS estimator. The sparsity-based estimator, LSE-SMP, improves the NMSE by 2dB over the SVP, but its NMSE is still 7\,dB higher than the linear MMSE estimator. The randomized-MMSE \cite{ Li2020Randomized}, which relies on randomized matrix approximation to reduce the complexity, is capable of achieving a comparable accuracy as the linear MMSE estimator. Observe that our rank-1 subspace method and its fast version attain the similar performance, {\color{black}both of which approximate the ML estimator and again outperform the linear MMSE estimator considerably.}

\begin{figure}[!t]
\centering
\includegraphics[width=62mm]{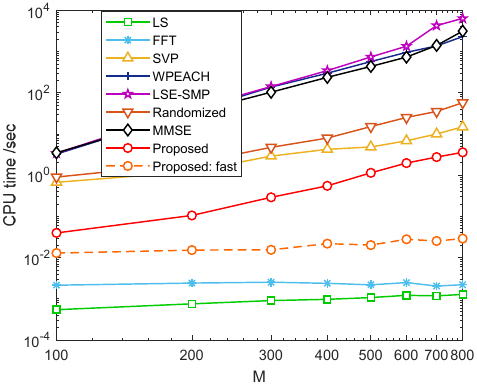}
\caption{Time complexity of various channel estimators as functions of $M$ for the low-rank sparse channel of Fig.~\ref{fig:5}, given $M=256$, $K=40$, $L=M/2$, and $B=K$.}
\label{fig:6}
\vspace*{-4mm}
\end{figure}

\begin{figure}[!t]
	\centering
	\includegraphics[width=6.2cm]{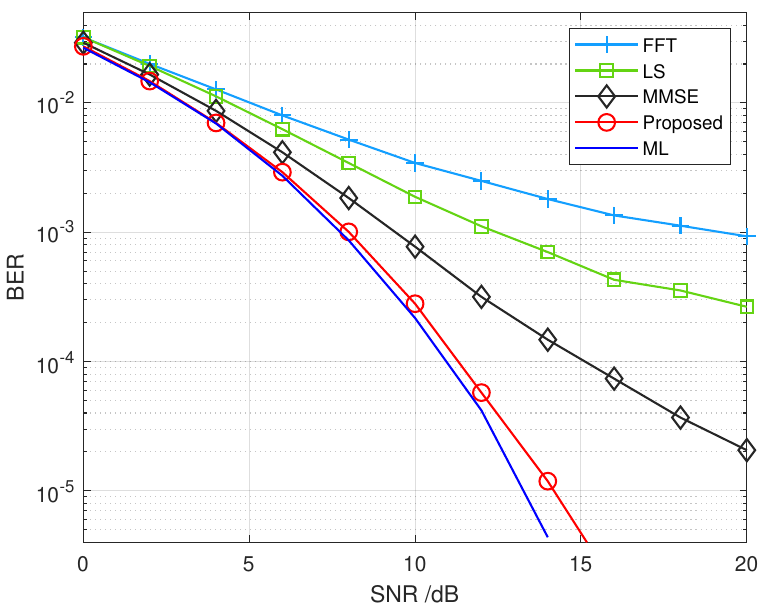}
	\caption{BER performance of various channel estimators in the 5G NR scenario; $M=256$, $K=40$ and $P=7$; $\text{rank}\{\H\}=K$; $B=80$; QPSK.}
	\label{fig:7}
	\vspace*{-1mm}
\end{figure}

Fig. \ref{fig:6} depicts the complexity of these channel estimators as the functions of $M$. The LSE-SMP requires the highest computational complexity, as it involves iterative computation of large matrices. The high complexity of the WPEACH can be attributed to the complex computation of weight coefficients required in approximating inverse via matrix series expansion. Observe that in correlated MIMO, the complexity of the linear MMSE method also increases considerably. By separately acquiring two small sub-matrices \cite{Li2020Randomized}, the randomized-MMSE has significantly lower complexity than the linear MMSE method, but its computational complexity is still considerably higher than our rank-1 subspace method. The LS and FFT methods offer the lowest complexity, but their estimation performance are much poorer compared with the linear MMSE method. Most noticeably, our fast rank-1 subspace method also requires a very low computational complexity. Specifically, at $M=300$, its time complexity is 10000$\times$ faster than the linear MMSE method. Among all the existing massive MIMO channel estimators, our new method is the first one that not only significantly outperforms the popular linear MMSE estimator, but also imposes a dramatically lower complexity than the latter.

\begin{figure}[!t]
\centering
\includegraphics[width=6cm]{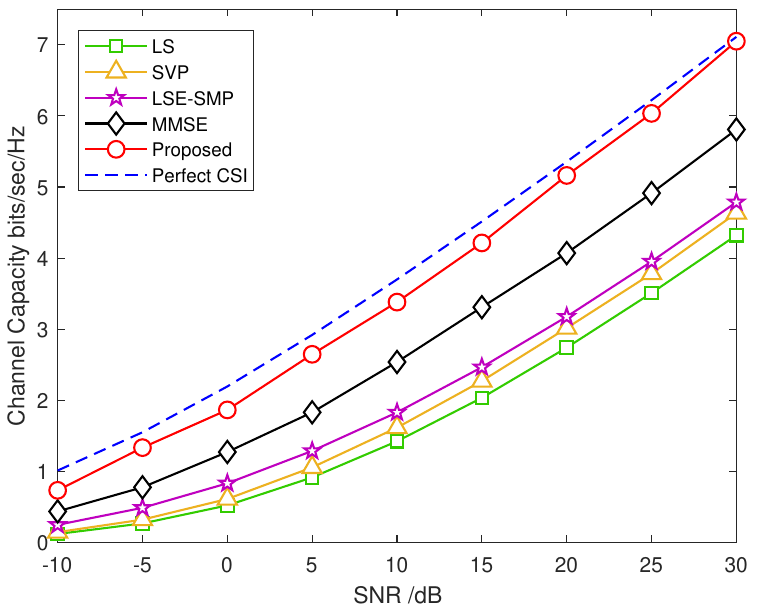}
\caption{Channel capacity achieved by various channel estimators for the MIMO system with $M=256$, $K=40$ and $P=4$.}
 \label{fig:8}
\vspace*{-1mm}
\end{figure}

\vspace{-0.4cm}

\subsection{BER Performance}
We further evaluate the bit error rate (BER) performance of various CSI estimators, in the context of 5G NR configurations.
To be specific, we consider the sparse uncorrelated channel matrix.
The detailed simulation parameters are given as follows: $M$ = 256, $K$ = 40, $B$ = $2K$ and $L$ = $M/2$.
The QPSK modulation is assumed, and the deterministic pilot sequence (i.e., DM-RS) of length $B=80$ is used.
In practice, our new estimator can be applied to both random and deterministic pilot sequences. In the simulation, totally $1.6\times10^6$ data symbols are used ($N_{\text{frame}}=10,~N_\text{packet}=2000$); the BER curve is then plotted based on the Monte-Carlo method.
From Fig. \ref{fig:7}, we can see that our new CSI estimator attains the near-ML BER performance, while outperforming the classical CSI estimators (e.g., LS and linear MMSE).

\subsection{Channel Capacity and Spectral Efficiency}\label{S6.3}
We then evaluate the channel capacity of the massive MIMO system, assuming $M=128$, $K=40$ and $P=7$.
For simplicity, we directly present the calculated results of the sum rate. One can refer to  \cite{yoo2006capacity} for the detailed equation.
Fig.~\ref{fig:8} depicts the channel capacity achieved by various channel estimators, in comparison with the channel capacity of the perfect CSI. It can be seen from Fig.~\ref{fig:8} that the channel capacity is significantly enhanced by our rank-1 subspace estimator, compared to the other methods. For example, at $\text{SNR}=30$\,dB, the achieved channel capacity of the linear MMSE estimator is 5.8\,bits/sec/Hz, while our new method improves the channel capacity to 7.1\,bits/sec/Hz, which asymptotically closes the \emph{capacity gap} between the estimated CSI and the perfect CSI (where the estimation MSE is zero).

Furthermore, we evaluate the spectral efficiency of various CSI estimators in massive MIMO communications.
In the simulation, we set $M=256$, $K=40$, $P=5$, $B=2K$ and ${\text{rank}}\{\H\}=K$.
Note that, with the finite number of antennas ($M=256$), we may assume that the inter-cell
interference (ICI) from nearby cells can be reasonably ignored,  as the pilot contamination in this
case can be negligible. Thus, we focus on the popular regularized zero-forcing (RZF) precoding scheme \cite{2018Massive}. On this basis, the signal-to-noise-plus-interference ratio (SINR) can be calculated.
Finally, we obtain the instantaneous spectral efficiency of various CSI estimators.
From Fig. \ref{fig:9}, we observe that our proposed method achieves comparable spectral efficiency as the ML estimator, while the others fail to do so.
\begin{figure}[!t]
	\centering
	\includegraphics[width=62mm]{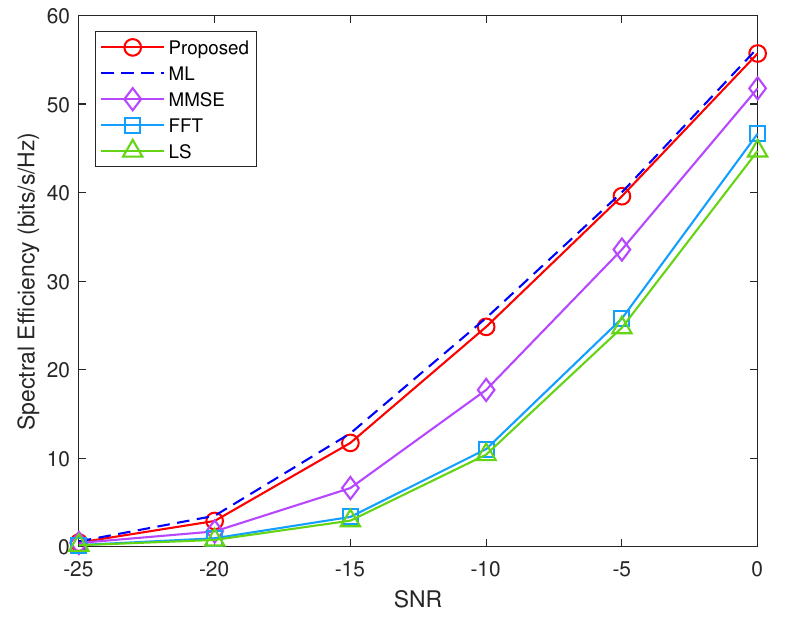}
\caption{Spectral efficiency of various channel estimators as functions of SNR; $M=256$, $K=40$, $L=M/2$, and $B=2K$. The regularized zero-forcing (RZF) precoding scheme is adopted.}
	\label{fig:9}
	%\vspace*{-4mm}
\end{figure}

\begin{figure}[!t]
%\vspace*{-3mm}
\centering
\includegraphics[width=6.2cm]{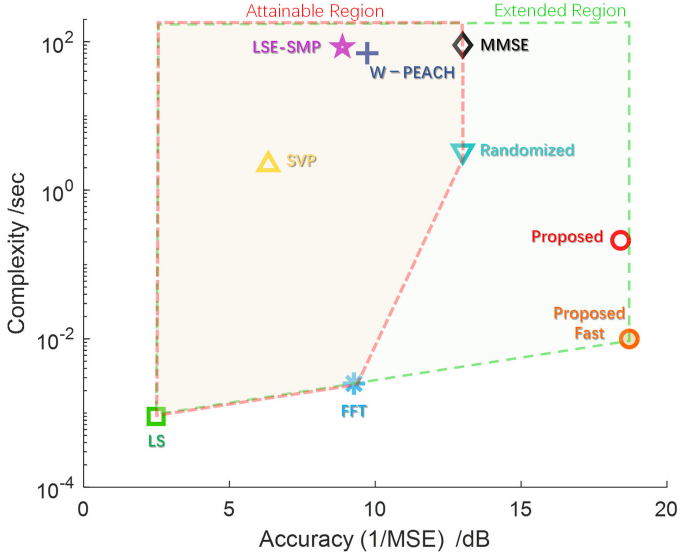}
\caption{Complexity vs accuracy in the CSI acquisition of massive MIMO communications, given $M=256$, $K=40$ and $\text{SNR}=20$\,dB. }
\label{fig:10}
\end{figure}

\begin{figure}[!t]
	\centering
	\includegraphics[width=6.2cm]{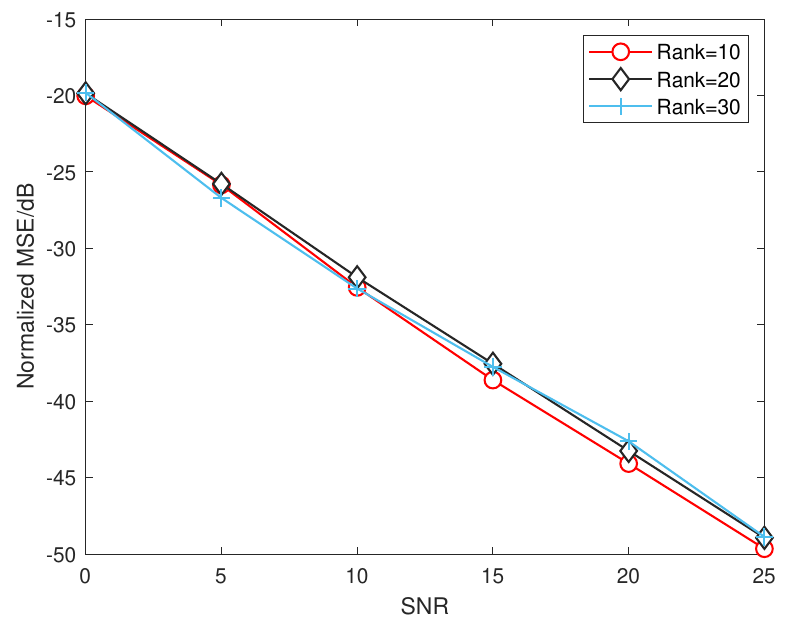}
	\caption{Performance of the proposed CSI estimator under different rank values of channel matrices; $M=256$, $K=40$ and $B=2K$. }
	\label{fig:11}
	\vspace*{-4mm}
\end{figure}

\vspace{-0.4cm}
\subsection{Complexity versus Accuracy Region}\label{S6.4}
Next we jointly evaluate the two aspects of channel estimation in massive MIMO communications, specifically, complexity vs accuracy.
The first performance metric is related to the implementation cost, such as time or energy; while the second metric corresponds to the benefit, i.e., estimation accuracy or channel capacity. Traditionally, these two conflicting aspects have to be balanced in massive MIMO systems, as done by different channel estimators. To achieve the high accuracy in the classical estimation sense, the complexity may become impractical, as in the case of the linear MMSE estimator. To reduce the complexity to a minimum, the estimation accuracy and hence the channel capacity has to be seriously scarified, as in the case of the LS estimator. Within the classical paradigm of trading such two aspects \cite{shen2015joint, shariati2014low,huang2019iterative,fan2017angle}, the attainable complexity-accuracy region in the CSI acquisition is largely limited, as illustrated in the attainable region of Fig.~\ref{fig:10}.

As a key breakthrough, our rank-1 subspace estimator surprisingly extends this complexity-accuracy region considerably, as seen in Fig.~\ref{fig:10}. Unlike the CRLB of the classical linear MMSE method whose MSE is proportional to the noise variance, our theoretical analysis shows that the attainable MSE of our new estimator is further scaled by $\frac{1}{M}$ when $M$ is large (e.g. $M>64$). Hence, the estimation accuracy of our rank-1 subspace estimator is far beyond the limit of the classical linear MMSE estimator, and it substantially enhances the channel capacity and the coverage area. Moreover, its fast version, which does not degrade the estimation accuracy, has the complexity scaled linearly with $M$. This new massive MIMO channel estimator therefore offers the great promise to the emerging 5G-Advanced and 6G communications.

%\vspace{-0.4cm}

\subsection{Practical Considerations}\label{S6.5}
\subsubsection{Effects of Channel Correlation Factor}
In the above, we have evaluated the performance of our proposed method in both full-rank and low-rank channels.
As seen, our method is independent of the rank value of a channel matrix.
It was shown in \cite{2007Simplified} that the channel correlation factor, or the spatial correlation matrices at the BS, affects the low rank property of massive MIMO channels.
In principle, the higher the channel correlation factor, the lower the channel rank value.
Although a unified relationship between the channel correlation factor and the channel rank value is hard to describe, we may alternatively evaluate the estimators' performance under different rank values, corresponding to various channel correlation factors.
In Fig. \ref{fig:11}, we present the NMSE performance with different rank values. From the averaged NMSE performance obtained with a large number of random channel realizations, we conclude that the proposed method accurately acquires unknown channel matrix.

\subsubsection{Effects of Sector-Level Beamforming}
In practice, the BS may adopt the sector-level beamforming to enhance the system's performance.
Combined with classical CSI estimators (e.g., LS), the sector-level beamforming may further improve the accuracy of channel estimation. Our proposed method can also be applied in this scenario.
Nevertheless, the improvement on the estimation accuracy may be marginal.
This is mainly due to the high-accuracy estimation of AoAs, as well as the subsequent post-reception beamforming in acquiring unknown fading gains based on the ML criterion.
Hence, for our two-stage CSI estimator, the MSE performance may be comparable with the case of using
omnidirectional antenna.

\subsubsection{2D Uniform Rectangular Array (URA) Scenario}
Furthermore, we demonstrate that the proposed channel estimator can be directly applied to the 2D URA based massive MIMO to accomplish a highly accurate channel estimation.
Here, the only difference is that the received signal becomes a two-dimensional matrix.
From Fig. \ref{fig:12}, it is seen that our new estimator attains significant performance gain over the classical LS and FFT estimators.

\subsubsection{Another Potential Solution}
		Lastly, we note that another popular scheme, i.e., approximate message passing (AMP), shares the similar linear inverse problem formulation. In this respect, the orthogonal AMP (OAMP) \cite{2017OAMP}, or recently developed memory AMP (MAMP) \cite{2022MAMP}, may also have the potential to accomplish the Bayesian optimal estimation. It is interesting to consider the application of such AMP methods to realize the high-accuracy CSI estimation in future research.
\begin{figure}[!t]
	\centering
	\includegraphics[width=6cm]{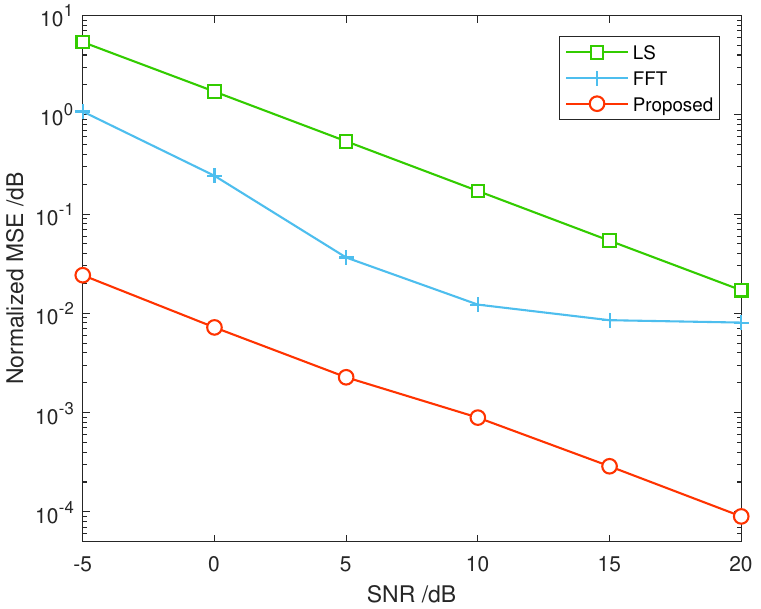}
	\caption{Normalized MSE performance of three channel estimators for 2D URA based MIMO, given $M=32\times 32$, $K=1$, $B=24$ and $P=5$.}
	\label{fig:12}
\end{figure}

\section{Conclusions}\label{S7}
For massive MIMO systems, the existing channel estimators do not trade off estimation accuracy and computational complexity well, and the attainable accuracy-complexity region is fundamentally limited. In this work, we have presented a novel rank-1 subspace channel estimator, which could approximate the ML method while avoiding the computationally exhaustive search of high-dimensional space. Our new estimator first acquires the high-resolution AoA estimate, and then attains the unbiased and highly accurate channel gains by the post-reception beamforming. Given the number of antennas $M$, our theoretical analysis has shown that this new method substantially outperforms the classical linear MMSE estimator, attaining the extra gain that increases with $\log_{10}(M)$. We have also derived a fast implementation of this rank-1 subspace estimator by leveraging the low-rank property of the constructed Hankel matrix. This fast version does not degrade the achievable performance, and yet it reduces the computational complexity to the order of $\textsf{O}(M)$, which is orders of magnitude lower than the linear MMSE estimator. Numerical simulations have validated our theoretical analysis. Our new estimator has thus resolved a long-standing dilemma in massive MIMO signal processing, namely, accuracy vs complexity. It is the first estimator that goes far beyond the estimation accuracy limit set by the linear MMSE estimator, and yet only imposes the complexity in the order of the LS estimator. We believe it has great potential in the 5G-Advanced and 6G communications.

%\section*{Acknowledgement}
%This work was supported by the Beijing Natural Science Foundation under Grant L234025, and in part by the National Natural Science Foundation of China under Grant 62071247.

\appendix
\subsection{Proof of Theorem~\ref{T1}}\label{ApA}
\begin{proof}
Substituting $\widehat{\theta}_p$ by $\theta_p$ in (\ref{eq5}) as well as noting the facts of $\bm{a}_M^{\rm H}(\theta_p)\bm{a}_M(\theta_p)=M$ and $\bm{a}_M^{\rm H}(\theta_p)\bm{a}_M(\theta_{p'})\rightarrow 0$ as $M\rightarrow \infty$ for $p\neq p'$, we immediately have
\begin{align}\label{eq:post_beam} % eqA1
\widehat{\alpha}_p =& \frac{1}{M} \bm{a}_M^{\rm H}(\theta_p)\bigg( \sum_{p'=0}^{P-1} \alpha_{p'} \bm{a}_M(\theta_{p'}) \bigg) + \frac{1}{M} \bm{a}_M^{\rm H}(\theta_p) \bm{n}, \nonumber \\
=& \alpha_p + n_M ,
\end{align}
where $\bm{n}\triangleq\bm{N}\bm{x}\in\mathbb{C}^{M\times 1}$ with elements having the variance $\sigma_n^2/B\sigma_x^2$, and $n_M\triangleq \bm{a}_M^{\rm H}(\theta_p) \bm{n}/M$ which is a zero-mean AWGN with the variance $\text{var}[n_M]=\sigma_n^2/MB\sigma_x^2$. Therefore, the estimated channel gains using (\ref{eq5}) are unbiased, i.e.,
\begin{align}\label{eqA2}
 \mathsf{E}\big[\widehat{\alpha}_p\big] =& \alpha_p, ~ 0\le p\le P-1 .
\end{align}
With the linear Gaussian model of (\ref{eq:post_beam}), the CRLB of the above channel estimator satisfies \cite{kay1993fundamentals}:
\begin{align} \label{eqA3}
\text{var}\big[\widehat{\alpha}_p\big] \geq & \text{var}[n_M] = \frac{\sigma_n^2}{M B\sigma_x^2}, ~ 0\le p\le P-1 .
\end{align}
This completes the proof of Theorem~\ref{T1}.
\end{proof}

\subsection{Proof of Theorem~\ref{T2}}\label{ApB}
\begin{proof}
For the ULA based massive MIMO system, the main-lobe gain of a post-reception beamformer is modeled as a circularly symmetric Gaussian function \cite{2006Reference,2010Channel,fan2018learning}, i.e.,
\begin{align}\label{eq:Gauss_beamformer} % eqB1
 G(\theta) =& M \, \exp\left( -\beta (\theta - \theta_p)^2 \right) ,
\end{align}
where $\beta \triangleq \ln(2)/\Delta\theta_{-3\text{dB}}^2$ with the 3dB beamwidth given by $\Delta\theta_{-3\text{dB}}=2/M$. Upon having the estimated AoAs $\big\{\widehat{\theta}_p\big\}_{p=0}^{P-1}$, the averaged channel gain estimation becomes:
\begin{align} \label{eqB2}
\mathsf{E}\Big[\widehat{\alpha}_p\Big] =& \exp\left( -2\ln(2)M^2\big(\widehat{\theta}_{p}-\theta_{p}\big)^2\right) \, \alpha_p .
\end{align}
When the  number of antennas $M$ is large, in order to attain an unbiased channel gain estimate, the AoA estimation error must satisfy certain condition. More specifically, in order to obtain $\mathsf{E}\Big[\widehat{\alpha}_p\Big]=\alpha_p$, we must have
\begin{align}\label{eq:B3}
M^2\big(\widehat{\theta}_{p}-\theta_{p}\big)^2  \overset{M~\text{is large}}{\longrightarrow} 0 .
\end{align}
It can be readily seen that when $\exists \,\varepsilon>0$ such that the AoA estimation error satisfies $\big|\widehat{\theta}_{p}-\theta_{p}\big| \thicksim \textsf{O}\big(1/M^{1+\varepsilon}\big)$, the unbiased channel gain estimate can be attained by (\ref{eq5}).
\end{proof}

\subsection{Proof of Theorem~\ref{T3}}\label{ApC}
\begin{proof}
From (\ref{eqB2}), the expectation on the estimated channel gain is given by
\begin{align}\label{eqC1}
& \mathsf{E}\Big[\widehat{\alpha}_{p,k}\Big] = \exp\left( -2\ln(2) M^2 \delta_{p,k}^2 \right) \alpha_{p,k} \nonumber \\
& \hspace*{5mm}= \bigg( 1 - 2\ln(2) \textsf{O}\bigg( \tfrac{\log_2(M)}{M^{2\epsilon}} \bigg) \bigg) \alpha_{p,k} \overset{M~\text{is large}}{\longrightarrow} \alpha_{p,k} .
\end{align}
The CRLB of the estimation error is defined by \cite{kay1993fundamentals}
\begin{align}\label{eqC2}
& \mathsf{E}\Big[\big|\widehat{\alpha}_{p,k} - \alpha_{p,k}\big|^2\Big] \geq - \left(\! \tfrac{\partial \ln p\left(\alpha_{p,k} - \frac{1}{M} \bm{a}_M^{\rm H}\big(\widehat{\theta}_{p,k}\big)\bm{y}_k\right)}{\partial \alpha_{p,k}^2} \! \right)^{-1} \!\! .
\end{align}
From (\ref{eq:Gauss_approximation}) and (\ref{eq:def_hk1}), we conclude that the MSE of the channel gain estimate meets:
\begin{align}\label{eqC3}
\mathsf{E}\Big[\big|\widehat{\alpha}_{p,k} - \alpha_{p,k}\big|^2\Big] \geq & \bar{\sigma}_M^2 \overset{M~\text{is large}}{\longrightarrow} \frac{\sigma_n^2}{MB\sigma_x^2}.
\end{align}
This completes the proof.
\end{proof}

\bibliographystyle{IEEEtran}
\bibliography{bib/IEEEabrv,bib/scibib}

\begin{thebibliography}{10}
\providecommand{\url}[1]{#1}
\csname url@samestyle\endcsname
\providecommand{\newblock}{\relax}
\providecommand{\bibinfo}[2]{#2}
\providecommand{\BIBentrySTDinterwordspacing}{\spaceskip=0pt\relax}
\providecommand{\BIBentryALTinterwordstretchfactor}{4}
\providecommand{\BIBentryALTinterwordspacing}{\spaceskip=\fontdimen2\font plus
\BIBentryALTinterwordstretchfactor\fontdimen3\font minus
  \fontdimen4\font\relax}
\providecommand{\BIBforeignlanguage}[2]{{%
\expandafter\ifx\csname l@#1\endcsname\relax
\typeout{** WARNING: IEEEtran.bst: No hyphenation pattern has been}%
\typeout{** loaded for the language `#1'. Using the pattern for}%
\typeout{** the default language instead.}%
\else
\language=\csname l@#1\endcsname
\fi
#2}}
\providecommand{\BIBdecl}{\relax}
\BIBdecl

\bibitem{marzetta2010noncooperative}
T.~L. Marzetta, ``Noncooperative cellular wireless with unlimited numbers of
  base station antennas,'' \emph{{IEEE} Trans. Wireless Commun.}, vol.~9,
  no.~11, pp. 3590--3600, Nov. 2010.

\bibitem{ngo2013energy}
H.~Q. Ngo, E.~G. Larsson, and T.~L. Marzetta, ``Energy and spectral efficiency
  of very large multiuser {MIMO} systems,'' \emph{{IEEE} Trans. Commun.},
  vol.~61, no.~4, pp. 1436--1449, Apr. 2013.

\bibitem{hoydis2013massive}
J.~Hoydis, S.~ten Brink, and M.~Debbah, ``Massive {MIMO} in the {UL/DL} of
  cellular networks: How many antennas do we need?'' \emph{{IEEE} J. Sel. Areas
  Commun.}, vol.~31, no.~2, pp. 160--171, Feb. 2013.

\bibitem{larsson2014massive}
E.~G. Larsson, O.~Edfors, F.~Tufvesson, and T.~L. Marzetta, ``Massive {MIMO}
  for next generation wireless systems,'' \emph{{IEEE} Commun. Mag.}, vol.~52,
  no.~2, pp. 186--195, Feb. 2014.

\bibitem{boccardi2014five}
F.~Boccardi, R.~W. Heath, A.~Lozano, T.~L. Marzetta, and P.~Popovski, ``{Five
  disruptive technology directions for 5G},'' \emph{{IEEE} Commun. Mag.},
  vol.~52, no.~2, pp. 74--80, Feb. 2014.

\bibitem{2015Fifty}
S.~Yang and L.~Hanzo, ``Fifty years of {MIMO} detection: The road to
  large-scale {MIMOs},'' \emph{{IEEE} Commun. Surveys Tuts.}, vol.~17, no.~4,
  pp. 1941--1988, Fourth Quarter 2015.

\bibitem{extreme_massive_MIMO_6G}
N.~Shlezinger, G.~C. Alexandropoulos, M.~F. Imani, Y.~C. Eldar, and D.~R.
  Smith, ``Dynamic metasurface antennas for {6G} extreme massive {MIMO}
  communications,'' \emph{{IEEE} Wireless Commun.}, vol.~28, no.~2, pp.
  106--113, Apr. 2021.

\bibitem{Wang_channel_model_6G}
J.~Wang, C.-X. Wang, J.~Huang, H.~Wang, and X.~Gao, ``A general {3D}
  space-time-frequency non-stationary {THz} channel model for {6G}
  ultra-massive {MIMO} wireless communication systems,'' \emph{{IEEE} J. Sel.
  Areas Commun.}, vol.~39, no.~6, pp. 1576--1589, Jun. 2021.

\bibitem{2016_mMIMO_HSR}
M.~Cheng, S.~Yang, and X.~Fang, ``Adaptive antenna-activation based beamforming
  for large-scale {MIMO} communication systems of high speed railway,''
  \emph{{China} Commun.}, vol.~13, no.~9, pp. 12--23, Sep. 2016.

\bibitem{2016_massive_MIMO_video}
S.~Yang, C.~Zhou, T.~Lv, and L.~Hanzo, ``Large-scale {MIMO} is capable of
  eliminating power-thirsty channel coding for wireless transmission of
  {HEVC/H.265} video,'' \emph{{IEEE} Wireless Commun.}, vol.~23, no.~3, pp.
  57--63, Jun. 2016.

\bibitem{shariati2014robust}
N.~Shariati, J.~Wang, and M.~Bengtsson, ``Robust training sequence design for
  correlated {MIMO} channel estimation,'' \emph{{IEEE} Trans. Signal Process.},
  vol.~62, no.~1, pp. 107--120, Jan. 2014.

\bibitem{Hu_ESPRIT_mMIMO}
A.~Hu, T.~Lv, H.~Gao, Z.~Zhang, and S.~Yang, ``An {ESPRIT}-based approach for
  {2-D} localization of incoherently distributed sources in massive {MIMO}
  systems,'' \emph{{IEEE} J. Sel. Topics Signal Process.}, vol.~8, no.~5, pp.
  996--1011, Oct. 2014.

\bibitem{LV201630}
T.~Lv, F.~Tan, H.~Gao, and S.~Yang, ``A beamspace approach for {2-D}
  localization of incoherently distributed sources in massive {MIMO} systems,''
  \emph{Signal Processing}, vol. 121, pp. 30--45, Apr. 2016.

\bibitem{Zhou_Angle_Estimation}
Y.~Zhou, Z.~Fei, S.~Yang, J.~Kuang, S.~Chen, and L.~Hanzo, ``Joint angle
  estimation and signal reconstruction for coherently distributed sources in
  massive {MIMO} systems based on {2-D} unitary esprit,'' \emph{IEEE Access},
  vol.~5, pp. 9632--9646, 2017.

\bibitem{semiblind_mMIMO_CE}
T.~Lv, S.~Yang, and H.~Gao, ``Semi-blind channel estimation relying on optimum
  pilots designed for multi-cell large-scale {MIMO} systems,'' \emph{IEEE
  Access}, vol.~4, pp. 1190--1204, 2016.

\bibitem{FDD_mMIMO_CE}
H.~Wang, G.~Li, S.~Zheng, S.~Yang, and P.~Pan, ``An approach to reduce the
  overhead of training sequences in {FDD} massive {MIMO} downlink systems,''
  \emph{{IEEE} Wireless Commun. Lett.}, vol.~8, no.~4, pp. 1301--1305, Aug.
  2019.

\bibitem{shariati2014low}
N.~Shariati, E.~Bj{\"o}rnson, M.~Bengtsson, and M.~Debbah, ``Low-complexity
  polynomial channel estimation in large-scale {MIMO} with arbitrary
  statistics,'' \emph{{IEEE} J. Sel. Topics Signal Process.}, vol.~8, no.~5,
  pp. 815--830, Oct. 2014.

\bibitem{Li2020Randomized}
B.~Li, S.~Wang, J.~Zhang, X.~Cao, and C.~Zhao, ``Randomized approximate channel
  estimator in massive-{MIMO} communication,'' \emph{{IEEE} Commun. Lett.},
  vol.~24, no.~10, pp. 2314 -- 2318, Oct. 2020.

\bibitem{Xie2017An}
H.~Xie, F.~Gao, and S.~Jin, ``An overview of low-rank channel estimation for
  massive {MIMO} systems,'' \emph{IEEE Access}, vol.~4, pp. 7313--7321, 2016.

\bibitem{shen2015joint}
W.~Shen, L.~Dai, B.~Shim, S.~Mumtaz, and Z.~Wang, ``Joint {CSIT} acquisition
  based on low-rank matrix completion for {FDD} massive {MIMO} systems,''
  \emph{{IEEE} Commun. Lett.}, vol.~19, no.~12, pp. 2178--2181, Dec. 2015.

\bibitem{eliasi2017low}
P.~A. Eliasi, S.~Rangan, and T.~S. Rappaport, ``Low-rank spatial channel
  estimation for millimeter wave cellular systems,'' \emph{{IEEE} Trans.
  Wireless Commun.}, vol.~16, no.~5, pp. 2748--2759, May 2017.

\bibitem{moshavi1996multistage}
S.~Moshavi, E.~Kanterakis, and D.~L. Schilling, ``Multistage linear receivers
  for {DS-CDMA} systems,'' \emph{International Journal of Wireless Information
  Networks}, vol.~3, pp. 1--17, Jan. 1996.

\bibitem{2010Compressed}
W.~U. Bajwa, J.~Haupt, A.~M. Sayeed, and R.~Nowak, ``Compressed channel
  sensing: A new approach to estimating sparse multipath channels,''
  \emph{Proc. IEEE}, vol.~98, no.~6, pp. 1058--1076, Jun. 2010.

\bibitem{alkhateeb2014channel}
A.~Alkhateeb, O.~E. Ayach, G.~Leus, and R.~W. Heath, ``Channel estimation and
  hybrid precoding for millimeter wave cellular systems,'' \emph{{IEEE} J. Sel.
  Topics Signal Process.}, vol.~8, no.~5, pp. 831--846, Oct. 2014.

\bibitem{alkhateeby2015compressed}
A.~Alkhateeb, G.~Leus, and R.~W. Heath, ``Compressed sensing based multi-user
  millimeter wave systems: How many measurements are needed?'' in \emph{Proc.
  2015 IEEE International Conference on Acoustics, Speech and Signal Processing
  (ICASSP)}, South Brisbane, Australia, Apr. 2015, pp. 2909--2913.

\bibitem{gao2016channel}
Z.~Gao, C.~Hu, L.~Dai, and Z.~Wang, ``Channel estimation for millimeter-wave
  massive {MIMO} with hybrid precoding over frequency-selective fading
  channels,'' \emph{{IEEE} Commun. Lett.}, vol.~20, no.~6, pp. 1259--1262, Jun.
  2016.

\bibitem{rodriguezfernandez2018frequency}
J.~Rodriguezfernandez, N.~Gonzalezprelcic, K.~Venugopal, and R.~W. Heath,
  ``Frequency-domain compressive channel estimation for frequency-selective
  hybrid millimeter wave {MIMO} systems,'' \emph{{IEEE} Trans. Wireless
  Commun.}, vol.~17, no.~5, pp. 2946--2960, May 2018.

\bibitem{huang2018asymptotically}
C.~Huang, L.~Liu, and C.~Yuen, ``Asymptotically optimal estimation algorithm
  for the sparse signal with arbitrary distributions,'' \emph{{IEEE} Trans.
  Veh. Technol.}, vol.~67, no.~10, pp. 10\,070--10\,075, Oct. 2018.

\bibitem{huang2019iterative}
C.~Huang, L.~Liu, C.~Yuen, and S.~Sun, ``Iterative channel estimation using
  {LSE} and sparse message passing for {mm-Wave MIMO} systems,'' \emph{{IEEE}
  Trans. Signal Process.}, vol.~67, no.~1, pp. 245--259, Jan. 2019.

\bibitem{krim1996two}
H.~Krim and M.~Viberg, ``{Two decades of array signal processing research: The
  parametric approach},'' \emph{IEEE Signal Processing Magazine}, vol.~13,
  no.~4, pp. 67--94, Jul. 1996.

\bibitem{schmidt1986multiple}
R.~Schmidt, ``{Multiple emitter location and signal parameter estimation},''
  \emph{{IEEE} Trans. Antennas Propag.}, vol.~34, no.~3, pp. 276--280, Mar.
  1986.

\bibitem{Li2019Fast}
B.~Li, S.~Wang, J.~Zhang, X.~Cao, and C.~Zhao, ``Ultra-fast accurate {AoA}
  estimation via automotive massive-{MIMO} radar,'' \emph{{IEEE} Trans. Veh.
  Technol.}, vol.~71, no.~2, pp. 1172--1186, Feb. 2022.

\bibitem{guo2017millimeter}
Z.~Guo, X.~Wang, and W.~Heng, ``Millimeter-wave channel estimation based on
  {2-D} beamspace {MUSIC} method,'' \emph{{IEEE} Trans. Wireless Commun.},
  vol.~16, no.~8, pp. 5384--5394, Aug. 2017.

\bibitem{fan2017angle}
D.~Fan, F.~Gao, G.~Wang, Z.~Zhong, and A.~Nallanathan, ``Angle domain signal
  processing-aided channel estimation for indoor 60{GHz TDD/FDD} massive {MIMO}
  systems,'' \emph{{IEEE} J. Sel. Areas Commun.}, vol.~35, no.~9, pp.
  1948--1961, Sep. 2017.

\bibitem{liao2016music}
W.~Liao and A.~Fannjiang, ``{MUSIC for single-snapshot spectral estimation:
  Stability and super-resolution},'' \emph{Applied and Computational Harmonic
  Analysis}, vol.~40, no.~1, pp. 33--67, Jan. 2016.

\bibitem{fortunati2014single}
S.~Fortunati, R.~Grasso, F.~Gini, M.~S. Greco, and K.~LePage, ``Single-snapshot
  {DOA} estimation by using compressed sensing,'' \emph{EURASIP Journal on
  Advances in Signal Processing}, vol. 2014, no. 120, pp. 1--17, Jul. 2014.

\bibitem{maisto2022single}
M.~A. Maisto, A.~Dell'Aversano, I.~Russo, A.~Brancaccio, and R.~Solimene, ``A
  single-snapshot {MUSIC} algorithm for {ADAS} radar processing,'' in
  \emph{Proc. 2022 Microwave Mediterranean Symposium (MMS)}, Pizzo Calabro,
  Italy, May 2022, pp. 1--6.

\bibitem{Brady2013Beamspace}
J.~Brady, N.~Behdad, and A.~M. Sayeed, ``Beamspace {MIMO} for millimeter-wave
  communications: System architecture, modeling, analysis, and measurements,''
  \emph{{IEEE} Trans. Antennas Propag.}, vol.~61, no.~7, pp. 3814--3827, Jul.
  2013.

\bibitem{2016Beamspace}
L.~Dai, X.~Gao, S.~Han, I.~Chih-Lin, and X.~Wang, ``Beamspace channel
  estimation for millimeter-wave massive {MIMO} systems with lens antenna
  array,'' in \emph{Proc. 2016 IEEE/CIC International Conference on
  Communications in China (ICCC)}, Chengdu, China, Jul. 2016, pp. 1--6.

\bibitem{Li2013On}
B.~Li, Z.~Zhou, W.~Zou, X.~Sun, and G.~Du, ``On the efficient beam-forming
  training for 60{GHz} wireless personal area networks,'' \emph{{IEEE} Trans.
  Wireless Commun.}, vol.~12, no.~2, pp. 504--515, Feb. 2013.

\bibitem{li2015a}
B.~Li, C.~Zhao, M.~Sun, H.~Zhang, Z.~Zhou, and A.~Nallanathan, ``A {Bayesian}
  approach for nonlinear equalization and signal detection in millimeter-wave
  communications,'' \emph{{IEEE} Trans. Wireless Commun.}, vol.~14, no.~7, pp.
  3794--3809, Jul. 2015.

\bibitem{gustafson2014on}
C.~Gustafson, K.~Haneda, S.~Wyne, and F.~Tufvesson, ``On {mm-Wave} multipath
  clustering and channel modeling,'' \emph{{IEEE} Trans. Antennas Propag.},
  vol.~62, no.~3, pp. 1445--1455, Mar. 2014.

\bibitem{akdeniz2014millimeter}
M.~R. Akdeniz, Y.~Liu, M.~K. Samimi, S.~Sun, S.~Rangan, T.~S. Rappaport, and
  E.~Erkip, ``Millimeter wave channel modeling and cellular capacity
  evaluation,'' \emph{{IEEE} J. Sel. Areas Commun.}, vol.~32, no.~6, pp.
  1164--1179, Jun. 2014.

\bibitem{1996Dynamical}
D.~Lai and G.~Chen, ``Dynamical systems identification from time-series data: A
  hankel matrix approach,'' \emph{Mathematical and Computer Modelling},
  vol.~24, no.~3, pp. 1--10, Aug. 1996.

\bibitem{2015Hankelet}
L.~Lo~Presti, M.~La~Cascia, S.~Sclaroff, and O.~Camps, ``{Hankelet-based
  dynamical systems modeling for 3D action recognition},'' \emph{Image and
  Vision Computing}, vol.~44, pp. 29--43, Dec. 2015.

\bibitem{georgakis2018dynamic}
C.~Georgakis, Y.~Panagakis, and M.~Pantic, ``{Dynamic behavior analysis via
  structured rank minimization},'' \emph{International Journal of Computer
  Vision}, vol. 126, pp. 333--357, Apr. 2018.

\bibitem{hacker2010single}
P.~Hacker and B.~Yang, ``{Single snapshot DOA estimation},'' \emph{Advances in
  Radio Science}, vol.~8, pp. 251--256, Oct. 2010.

\bibitem{roy1989esprit-estimation}
R.~Roy and T.~Kailath, ``{ESPRIT-estimation of signal parameters via rotational
  invariance techniques},'' \emph{{IEEE} Trans. Acoust., Speech, Signal
  Process.}, vol.~37, no.~7, pp. 984--995, Jul. 1989.

\bibitem{Li2021}
B.~Li, S.~Wang, Z.~Feng, J.~Zhang, X.~Cao, and C.~Zhao, ``Fast pseudospectrum
  estimation for automotive massive {MIMO} radar,'' \emph{IEEE Internet of
  Things Journal}, vol.~8, no.~20, pp. 15\,303--15\,316, Oct. 2021.

\bibitem{woodruff2014sketching}
D.~P. Woodruff, ``Sketching as a tool for numerical linear algebra,''
  \emph{Foundations and Trends® in Theoretical Computer Science}, vol.~10, no.
  1–2, pp. 1--157, 2014.

\bibitem{wang2016spsd}
S.~Wang, L.~Luo, and Z.~Zhang, ``{SPSD} matrix approximation vis column
  selection: Theories, algorithms, and extensions,'' \emph{Journal of Machine
  Learning Research}, vol.~17, no.~49, pp. 1697--1745, May 2016.

\bibitem{drineas2005on}
P.~Drineas and M.~W. Mahoney, ``On the {Nystr$\ddot{\text{o}}$m} method for
  approximating a {Gram} matrix for improved kernel-based learning,''
  \emph{Journal of Machine Learning Research}, vol.~6, no.~72, pp. 2153--2175,
  Dec. 2005.

\bibitem{stoica1989music}
P.~Stoica and A.~Nehorai, ``{MUSIC}, maximum likelihood, and {Cramer-Rao}
  bound,'' \emph{{IEEE} Trans. Acoust., Speech, Signal Process.}, vol.~37,
  no.~5, pp. 720--741, May 1989.

\bibitem{kay1993fundamentals}
S.~M. Kay, \emph{Fundamentals of Statistical Signal Processing: Estimation
  Theory}.\hskip 1em plus 0.5em minus 0.4em\relax USA: Prentice-Hall, 1993.

\bibitem{dong2002optimal}
M.~Dong and L.~Tong, ``Optimal design and placement of pilot symbols for
  channel estimation,'' \emph{{IEEE} Trans. Signal Process.}, vol.~50, no.~12,
  pp. 3055--3069, Dec. 2002.

\bibitem{berriche2004cramer}
L.~Berriche, K.~Abed-Meraim, and J.-C. Belfiore, ``{Cramer-Rao} bounds for
  {MIMO} channel estimation,'' in \emph{Proc. 2004 IEEE International
  Conference on Acoustics, Speech, and Signal Processing (ICASSP)}, vol.~IV,
  Montreal, Canada, May 2004, pp. 397--400.

\bibitem{2021Random}
B.~Li, P.~Chen, H.~Liu, W.~Guo, X.~B. Cao, J.~Z. Du, C.~L. Zhao, and J.~Zhang,
  ``Random sketch learning for deep neural networks in edge computing,''
  \emph{Nature Computational Science}, vol.~1, pp. 221--228, Mar. 2021.

\bibitem{2023Machine}
B.~Li, Z.~P. Wei, J.~T. Wu, S.~Yu, T.~Zhang, C.~L. Zhu, D.~Z. Zheng, W.~S. Guo,
  C.~L. Zhao, and J.~Zhang, ``Machine learning-enabled globally guaranteed
  evolutionary computation,'' \emph{Nature Machine Intelligence}, vol.~5,
  no.~4, p. 457–467, 2023.

\bibitem{gittens2016revisiting}
A.~Gittens and M.~W. Mahoney, ``Revisiting the {Nystr$\ddot{\text{o}}$m} method
  for improved large-scale machine learning,'' \emph{Journal of Machine
  Learning Research}, vol.~17, no. 117, pp. 3977--4041, Apr. 2016.

\bibitem{balevi2020massive}
E.~Balevi, A.~Doshi, and J.~G. Andrews, ``{Massive MIMO channel estimation with
  an untrained deep neural network},'' \emph{{IEEE} Trans. Wireless Commun.},
  vol.~19, no.~3, pp. 2079--2090, Mar. 2020.

\bibitem{yoo2006capacity}
T.~Yoo and A.~Goldsmith, ``{Capacity and power allocation for fading MIMO
  channels with channel estimation error},'' \emph{{IEEE} Trans. Inf. Theory},
  vol.~52, no.~5, pp. 2203--2214, May 2006.

\bibitem{2018Massive}
N.~Fatema, G.~Hua, Y.~Xiang, D.~Peng, and I.~Natgunanathan, ``Massive mimo
  linear precoding: A survey,'' \emph{IEEE Systems Journal}, pp. 1--12, 2018.

\bibitem{2007Simplified}
A.~Forenza, D.~J. Love, and R.~W. Heath, ``Simplified spatial correlation
  models for clustered mimo channels with different array configurations,''
  \emph{IEEE Trans. Veh. Technol.}, vol.~56, pp. 1924--1934, 2007.

\bibitem{2017OAMP}
J.~J. Ma and P.~Li, ``Orthogonal {AMP},'' \emph{IEEE Access}, pp. 2020--2033,
  2017.

\bibitem{2022MAMP}
L.~Liu, S.~Huang, and B.~M. Kurkoski, ``Memory {AMP},'' \emph{IEEE Transactions
  on Information Theory}, pp. 8015--8039, 2022.

\bibitem{2006Reference}
I.~Toyoda, \emph{{Reference antenna model with side lobe for TG3c evaluation}},
  IEEE Std. IEEE 802.15-06-0474-00-003c, 2006.

\bibitem{2010Channel}
A.~Maltsev, \emph{{Channel models for 60GHz WLAN systems}}, IEEE Std. IEEE
  802.11-09/0334r8, 2010.

\bibitem{fan2018learning}
C.~Fan, B.~Li, C.~Zhao, W.~Guo, and Y.~Liang, ``Learning-based spectrum sharing
  and spatial reuse in {mm-Wave} ultradense networks,'' \emph{{IEEE} Trans.
  Veh. Technol.}, vol.~67, no.~6, pp. 4954--4968, Jun. 2018.

\end{thebibliography}

\begin{IEEEbiography}[{\includegraphics[width=1in,height=1.25in,clip,keepaspectratio]{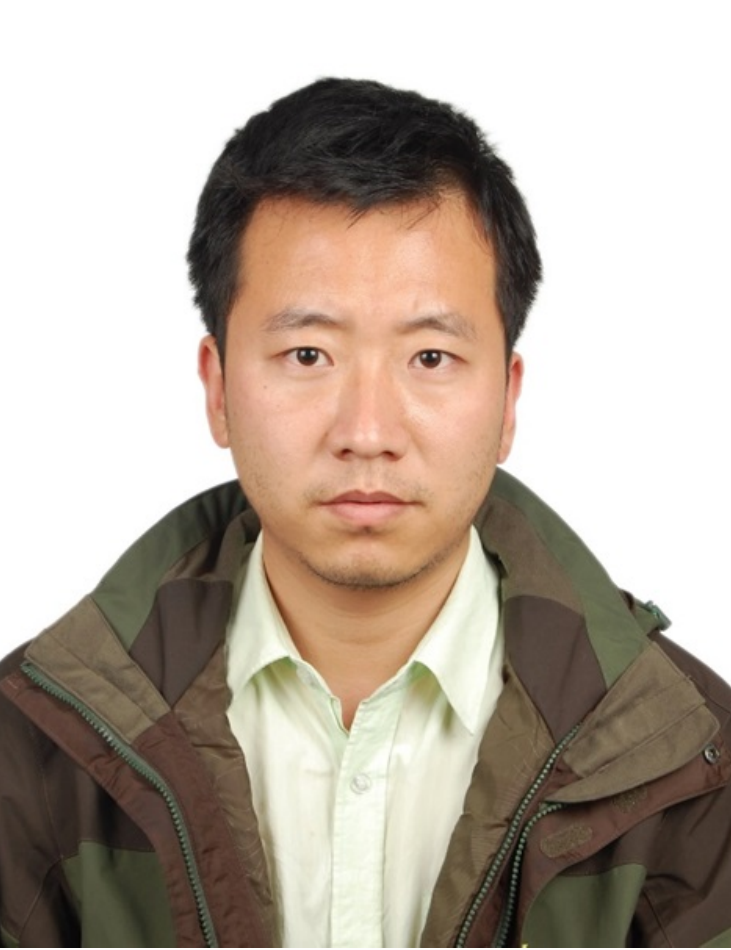}}]{Bin Li}
received the B.Eng. degree in electrical information engineering from Beijing University of Chemical Technology in 2007, and the Ph.D. degree in communication and information engineering from Beijing University of Posts and Telecommunications (BUPT) in 2013. In the same year, he joined BUPT, where he is currently an Associate Professor with the School of Information and Communication Engineering. His current research interests include signal processing for wireless communications and machine learning, such as millimeter-wave communications, UAV communications, MIMO communication/radar systems. He received the 2011 ChinaCom Best Paper Award, the 2015 IEEE WCSP Best Paper Award. He is an Associate Editor of \textit{China Communications}.
\end{IEEEbiography}

\begin{IEEEbiography}[{\includegraphics[width=1in,height=1.25in,clip,keepaspectratio]{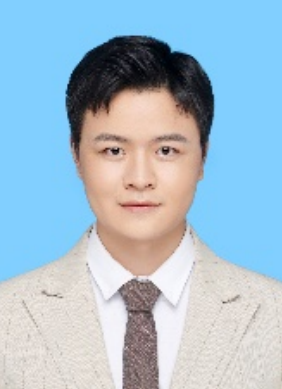}}]{Ziping Wei} 
received the B.Eng. degree in communication engineering from Civil Aviation University of China (CAUC) in 2018. He is currently pursuing the Ph.D. degree with the Key Laboratory of Universal Wireless Communications, Ministry of Education, BUPT, China. His research interests include MIMO signal processing, channel estimation, and radar signal processing.
\end{IEEEbiography}

\begin{IEEEbiography}[{\includegraphics[width=1in,height=1.25in,clip,keepaspectratio]{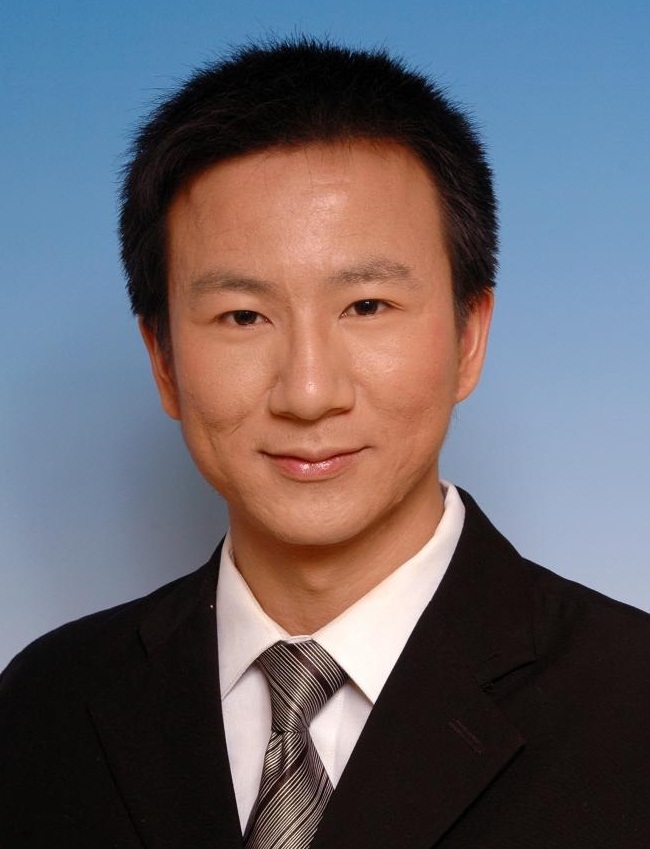}}]{Shaoshi Yang} (Senior Member, IEEE)  
received the B.Eng. degree in information engineering from Beijing University of Posts and Telecommunications (BUPT), China, in 2006, and the Ph.D. degree in electronics and electrical engineering from University of Southampton, UK, in 2013. From 2008 to 2009, he was a Researcher with Intel Labs China. From 2013 to 2016, he was a Research Fellow with the School of Electronics and Computer Science, University of Southampton. From 2016 to 2018, he was a Principal Engineer with Huawei Technologies Co., Ltd., where he made significant contributions to Huawei’s products and solutions on 5G base stations, wideband IoT, and cloud gaming/VR. He is currently a Full Professor with BUPT. His research interests include 5G/5G-A/6G, massive MIMO, mobile ad hoc networks, distributed artificial intelligence, and cloud gaming/VR. He received numerous research awards from University of Southampton, Huawei., IEEE ComSoc, Xiaomi Foundation, China Association of Inventions, and China Industry-University-Research Institute Collaboration Association. He was/is an Editor for \textit{IEEE Systems Journal}, \textit{IEEE Wireless Communications Letters}, and \textit{Signal Processing} (Elsevier). He is a standing committee member of the China Computer Federation (CCF) Technical Committee on Distributed Computing and Systems. He was also a Guest Researcher with the Isaac Newton Institute for Mathematical Sciences, Cambridge University. For more details of his research progress, please refer to https://shaoshiyang.weebly.com/ 
\end{IEEEbiography}

\begin{IEEEbiography}[{\includegraphics[width=1in,height=1.25in,clip,keepaspectratio]{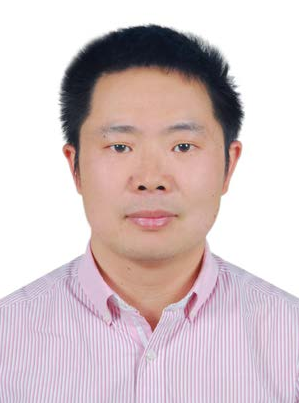}}]{Yang Zhang} received the B.Eng. degree in communication engineering in 2007, and the M.Eng. degree in computer applied technology in 2010, both from China University of Petroleum (CUP), and the Ph.D. degree in communication and information engineering from Beijing University of Posts and Telecommunications (BUPT) in 2014. He is currently with China Satellite Network Innovation Research Institute Co., Ltd. His research interests include MIMO signal processing, satellite communications and networking.  
\end{IEEEbiography}

\begin{IEEEbiography}[{\includegraphics[width=1in,height=1.25in,clip,keepaspectratio]{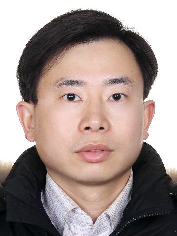}}]{Jun Zhang} (Senior Member, IEEE)  
received the M.S. degree in statistics with the Department of Mathematics from Southeast University, Nanjing, China, in 2009, and the Ph.D. degree in communication and information systems with the National Mobile Communications Research Laboratory, Southeast University, Nanjing, China, in 2013. From 2013 to 2015, he was a Postdoctoral Research Fellow with Singapore University of Technology and Design, Singapore. Since 2015, he is on the faculty of the Jiangsu Key Laboratory of Wireless Communications, School of Communications and Information Engineering, Nanjing University of Posts and Telecommunications, where he is currently a Professor. His research interests include massive MIMO communications, RIS-assisted wireless communications, UAV-assisted wireless communications, physical layer security, and large dimensional random matrix theory. Dr. Zhang was a recipient of the IEEE GLOBECOM Best Paper Award in 2016, the IEEE APCC Best Paper Award in 2017, the IEEE JC\&S Symposium Best Paper Award in 2022, the IEEE/CIC ICCC Best Paper Award in 2023, and the WCSP Best Paper Award in 2023. He was an Associate Editor for \textit{IEEE Communications Letters.}
\end{IEEEbiography}

\begin{IEEEbiography}[{\includegraphics[width=1in,height=1.25in,clip,keepaspectratio]{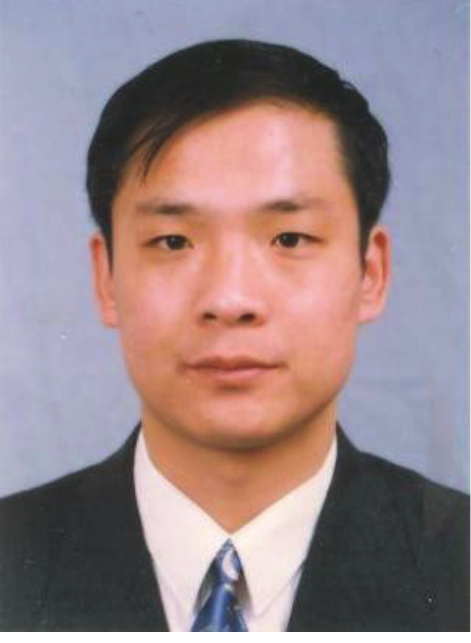}}]{Chenglin Zhao} 
received the B.Eng. degree in radio technology from Tianjin University in 1986, and the M.Eng. degree in circuits and systems from Beijing University of Posts and Telecommunications (BUPT) in 1993, and the Ph.D. degree in communication and information systems from BUPT, in 1997. At present, he serves as a Professor in BUPT. His research is focused on emerging technologies of short-range wireless communications, cognitive radio, 60GHz millimeter-wave communications.
\end{IEEEbiography}

\begin{IEEEbiography}[{\includegraphics[width=1in,height=1.25in,clip,keepaspectratio]{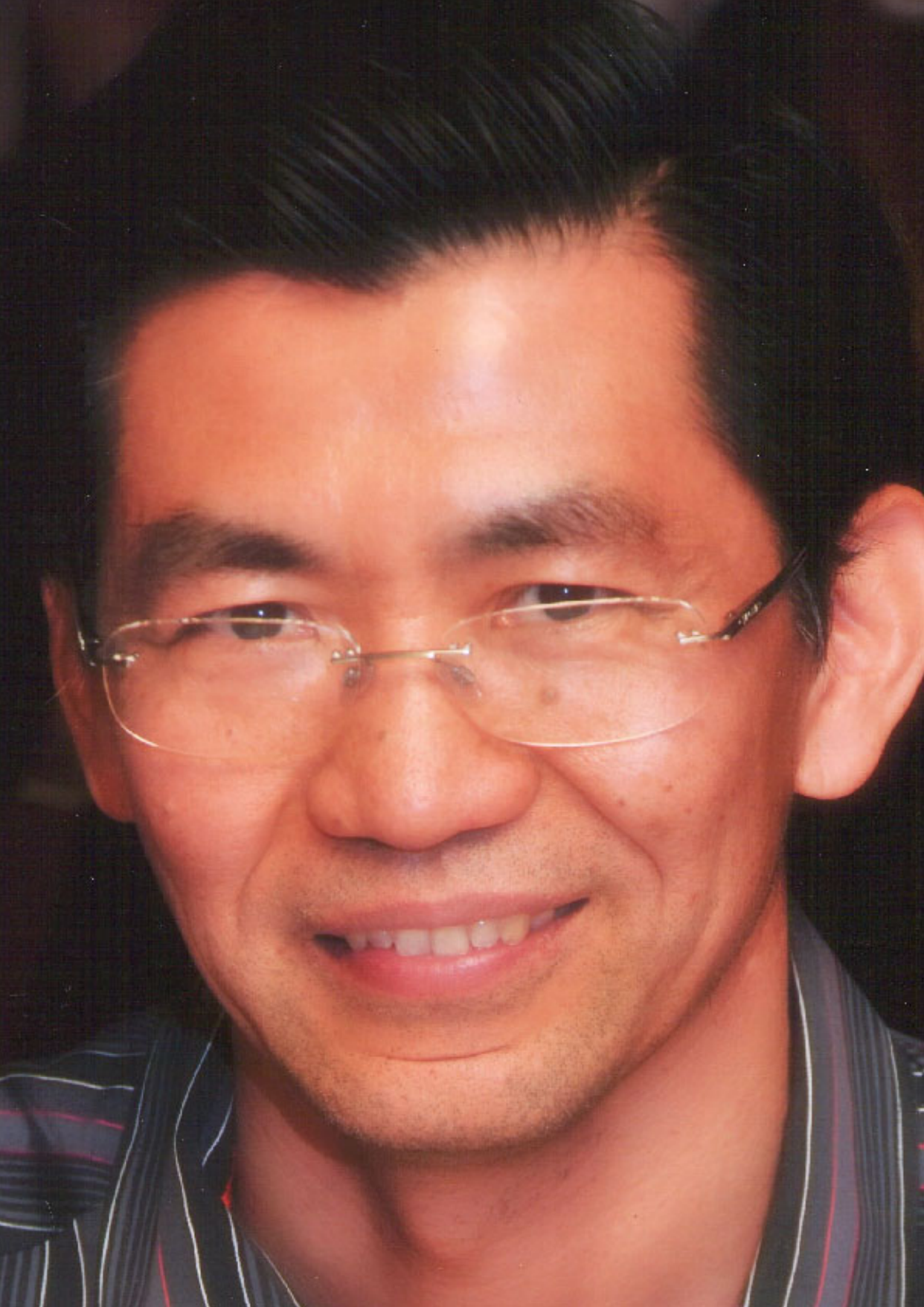}}]{Sheng Chen} (Life Fellow, IEEE) received his B.Eng. degree from East China Petroleum Institute, Dongying, China, in 1982, and his Ph.D. degree from City, University of London, in 1986, both in control engineering. In 2005, he was awarded the higher doctoral degree, Doctor of Sciences (DSc), from University of Southampton, UK. From 1986 to 1999, He held research and academic appointments at University of Sheffield, University of Edinburgh and University of Portsmouth, all in UK. Since 1999, he has been with the School of Electronics and Computer Science, University of Southampton, where he holds the post of Professor in Intelligent Systems and Signal Processing. Dr. Chen's research interests include machine learning, neural networks and wireless communications. He has published over 700 research papers. He has 19,000+ Web of Science citations with h-index 61 and 38,000+ Google Scholar citations with h-index 83. He is also a Fellow of the Royal Academy of Engineering, UK, a Fellow of Asia-Pacific Artificial Intelligence Association (AAIA), and a Fellow of IET. He is one of the original ISI highly cited researcher in engineering (March 2004).
\end{IEEEbiography}

\end{document}